\documentclass[pra,10pt,amsmath,amssymb,nofootinbib,bm,bbm,notitlepage,superscriptaddress]{revtex4-1}
\usepackage[english]{babel}
\usepackage[utf8]{inputenc}
\usepackage{xcolor}

\setcitestyle{sort&compress,numbers}
\usepackage{amsthm, color, mathtools, amsmath, amssymb, setspace,bbm, tensor, subfigure, mathtools, placeins, graphicx, wrapfig,float, verbatim, xcolor,enumitem}
\usepackage{graphicx}
\usepackage[unicode=true, breaklinks=false, pdfborder={0 0 1}, backref=false, colorlinks=true, linkcolor=blue, citecolor=blue]{hyperref}
\usepackage{cancel,bm}
\usepackage{float}

\newcommand{\calcium}{$^{40}\rm{Ca}^+$}

\AtBeginDocument{\providecommand\Eqref[1]{\ref{Eq:#1}}}
\usepackage[T1]{fontenc}
\usepackage{geometry}
\usepackage{refstyle}
\usepackage{cancel}
\usepackage{graphicx}
\PassOptionsToPackage{normalem}{ulem}
\usepackage{ulem}

\makeatletter

\AtBeginDocument{\providecommand\Assuref[1]{\ref{Assu:#1}}}
\AtBeginDocument{\providecommand\Eqref[1]{\ref{Eq:#1}}}
\AtBeginDocument{\providecommand\Figref[1]{\ref{Fig:#1}}}
\AtBeginDocument{\providecommand\Algref[1]{\ref{Alg:#1}}}
\AtBeginDocument{\providecommand\Lemref[1]{\ref{Lem:#1}}}
\AtBeginDocument{\providecommand\Defref[1]{\ref{Def:#1}}}
\AtBeginDocument{}
\RS@ifundefined{subsecref}
  {\newref{subsec}{name = Subsection \RSsectxt}}
  {}
\RS@ifundefined{thmref}
  {\def\RSthmtxt{theorem~}\newref{thm}{name = \RSthmtxt}}
  {}
\RS@ifundefined{lemref}
  {\def\RSlemtxt{lemma~}\newref{lem}{name = \RSlemtxt}}
  {}

\numberwithin{equation}{section}
\numberwithin{figure}{section}
\theoremstyle{plain}
\newtheorem{thm}{\protect\theoremname}
\theoremstyle{definition}
\newtheorem{defn}[thm]{\protect\definitionname}
\theoremstyle{plain}
\newtheorem{lyxalgorithm}[thm]{\protect\algorithmname}
\theoremstyle{plain}
\newtheorem{lem}[thm]{\protect\lemmaname}
\theoremstyle{remark}
\newtheorem{rem}[thm]{\protect\remarkname}

\theoremstyle{plain}
\newtheorem{assumption}[thm]{\protect\assumptionname}

\newref{def}{name=definition~,Name=Definition~,names=definitions~,Names=Definitions~}
\newref{lem}{name=lemma~,Name=Lemma~,names=lemmas~,Names=Lemmas~}
\newref{sec}{name=section~,Name=Section~,names=sections~,Names=Sections~}
\newref{subsec}{name=subsection~,Name=Subsection~,names=subsections~,Names=Subsections~}
\newref{rem}{name=remark~,Name=Remark~,names=remarks~,Names=Remarks~}
\newref{eq}{name=equation~,Name=Equation~,names=equations~,Names=Equations~}
\newref{alg}{name=algorithm~,Name=Algorithm~,names=algorithms~,Names=Algorithms~}
\newref{assu}{name=assumption~,Name=Assumption~,names=assumptions~,Names=Assumptions~}

\makeatother

\usepackage{babel}
\providecommand{\algorithmname}{Algorithm}
\providecommand{\assumptionname}{Assumption}
\providecommand{\definitionname}{Definition}
\providecommand{\lemmaname}{Lemma}
\providecommand{\remarkname}{Remark}
\providecommand{\theoremname}{Theorem}

\def\be{\begin{equation}}
\def\ee{\end{equation}}
\newcommand{\bra}[1]{\langle #1|}
\newcommand{\ket}[1]{|#1 \rangle}

\newcommand{\tr}{\mathrm{Tr}}

\newcommand{\subfigimg}[3][,]{%
	\setbox1=\hbox{\includegraphics[#1]{#3}}%
	\leavevmode\rlap{\usebox1}%
	\rlap{\hspace*{2pt}\raisebox{\dimexpr\ht1-0.5\baselineskip}{{\bfseries \large\textsf{#2}}}}%
	\phantom{\usebox1}%
}

\usepackage{color}
\definecolor{KB}{rgb}{0.4,0.3,0.9}

\makeatother

\usepackage{babel}
\begin{document}
\title{Self-Testing of a Single Quantum System: Theory and Experiment}

\author{Xiao-Min Hu}
\email{These five authors contributed equally to this work.}
\affiliation{CAS Key Laboratory of Quantum Information, University of Science and Technology of China, Hefei, 230026, People's Republic of China}
\affiliation{CAS Center For Excellence in Quantum Information and Quantum Physics, University of Science and Technology of China, Hefei, 230026, People's Republic of China}
\author{Yi Xie}
\email{These five authors contributed equally to this work.}
\affiliation{Department of Physics, College of Liberal Arts and Sciences, National University of
Defense Technology, Changsha 410073, Hunan, China}
\affiliation{Interdisciplinary Center for Quantum Information, National University of Defense
Technology, Changsha 410073, Hunan, China}
\author{Atul Singh Arora}
\email{These five authors contributed equally to this work.}
\affiliation{Université libre de Bruxelles, Belgium}
\affiliation{Institute of Quantum Information and Matter, Department of Computing and Mathematical Sciences, California Institute of Technology, USA}
\author{Ming-Zhong Ai}
\email{These five authors contributed equally to this work.}
\affiliation{CAS Key Laboratory of Quantum Information, University of Science and Technology of China, Hefei, 230026, People's Republic of China}
\affiliation{CAS Center For Excellence in Quantum Information and Quantum Physics, University of Science and Technology of China, Hefei, 230026, People's Republic of China}

\author{Kishor Bharti}
\email{These five authors contributed equally to this work.}
\affiliation{Centre for Quantum Technologies, National University of Singapore, Singapore}
\affiliation{Joint Center for Quantum Information and Computer Science and Joint Quantum Institute, NIST/University of Maryland, College Park, Maryland 20742, USA}

\author{Jie Zhang}
\affiliation{Department of Physics, College of Liberal Arts and Sciences, National University of
Defense Technology, Changsha 410073, Hunan, China}
\affiliation{Interdisciplinary Center for Quantum Information, National University of Defense
Technology, Changsha 410073, Hunan, China}
\author{Wei Wu}
\affiliation{Department of Physics, College of Liberal Arts and Sciences, National University of
Defense Technology, Changsha 410073, Hunan, China}
\affiliation{Interdisciplinary Center for Quantum Information, National University of Defense
Technology, Changsha 410073, Hunan, China}
\author{Ping-Xing Chen}
\affiliation{Department of Physics, College of Liberal Arts and Sciences, National University of
Defense Technology, Changsha 410073, Hunan, China}
\affiliation{Interdisciplinary Center for Quantum Information, National University of Defense
Technology, Changsha 410073, Hunan, China}
\author{Jin-Ming Cui}
\affiliation{CAS Key Laboratory of Quantum Information, University of Science and Technology of China, Hefei, 230026, People's Republic of China}
\affiliation{CAS Center For Excellence in Quantum Information and Quantum Physics, University of Science and Technology of China, Hefei, 230026, People's Republic of China}
\author{Bi-Heng Liu}
\affiliation{CAS Key Laboratory of Quantum Information, University of Science and Technology of China, Hefei, 230026, People's Republic of China}
\affiliation{CAS Center For Excellence in Quantum Information and Quantum Physics, University of Science and Technology of China, Hefei, 230026, People's Republic of China}
\author{Yun-Feng Huang}
\affiliation{CAS Key Laboratory of Quantum Information, University of Science and Technology of China, Hefei, 230026, People's Republic of China}
\affiliation{CAS Center For Excellence in Quantum Information and Quantum Physics, University of Science and Technology of China, Hefei, 230026, People's Republic of China}
\author{Chuan-Feng Li}
\affiliation{CAS Key Laboratory of Quantum Information, University of Science and Technology of China, Hefei, 230026, People's Republic of China}
\affiliation{CAS Center For Excellence in Quantum Information and Quantum Physics, University of Science and Technology of China, Hefei, 230026, People's Republic of China}
\author{Guang-Can Guo}
\affiliation{CAS Key Laboratory of Quantum Information, University of Science and Technology of China, Hefei, 230026, People's Republic of China}
\affiliation{CAS Center For Excellence in Quantum Information and Quantum Physics, University of Science and Technology of China, Hefei, 230026, People's Republic of China}

\author{Jérémie Roland}
\affiliation{Université libre de Bruxelles, Belgium}

\author{Ad\'{a}n Cabello}
\affiliation{Departamento de F\'{i}sica Aplicada II, Universidad de Sevilla, E-41012 Sevilla, Spain}
\affiliation{Instituto Carlos I de F\'{\i}sica Te\'orica y Computacional, Universidad de Sevilla, E-41012 Sevilla, Spain}

\author{Leong-Chuan Kwek}
\affiliation{Centre for Quantum Technologies, National University of Singapore, Singapore} \affiliation{MajuLab, CNRS-UNS-NUS-NTU International Joint Research Unit, Singapore UMI 3654, Singapore}
\affiliation{National Institute of Education, Nanyang Technological University, Singapore 637616, Singapore}

\begin{abstract}
Certifying individual quantum devices with minimal assumptions is crucial for the development of quantum technologies. Here, we investigate how to leverage single-system contextuality to realize self-testing. We develop a robust self-testing protocol based on the simplest contextuality witness for the simplest contextual quantum system, the Klyachko-Can-Binicioğlu-Shumovsky (KCBS) inequality for the qutrit. We establish a lower bound on the fidelity of the state and the measurements (to an ideal configuration) as a function of the value of the witness under a pragmatic assumption on the measurements we call the \emph{KCBS orthogonality} condition. We apply the method in an experiment with randomly chosen measurements on a single trapped $^{40}{\rm Ca}^+$ and near-perfect detection efficiency. The observed statistics allow us to self-test the system and provide the first experimental demonstration of quantum self-testing of a single system. Further, we quantify and report that deviations from our assumptions are minimal, an aspect previously overlooked by contextuality experiments. 

\end{abstract}

\maketitle

\section{Introduction}
We all believe that quantum devices typically have an edge over its classical analogues.  Yet, certification of quantum devices remains challenging and difficult with increased dimensionality.  
As system size increases, there is an exponential scaling in the Hilbert space, making certification hard \cite{Wang285,Gongeabg7812,PhysRevLett.106.130506}.   How do we confidently say that a quantum device is really a quantum device? Developing such confidence needs certification techniques that require minimal assumptions about their operation. 

Self-testing constitutes an effective method for providing assurance regarding the internal workings of quantum devices using on experimental data, but subject to appropriate assumptions \cite{vsupic2020self}. Mayers and Yao published a seminal study in 2004 that relied on Bell nonlocality to create the first self-testing methodology for a pair of non-communicating entangled quantum devices \cite{Yao_self}.  Indeed, the concept of self-testing can be traced back to earlier works of Tsirelson~\cite{tsirel1987quantum}, Summers-Werner~\cite{summers1987bell}, and Popescu-Rohrlich~\cite{popescu1992generic}.  Self-testing via Bell nonlocality has been demonstrated for all pure bipartite entangled states~\cite{coladangelo2017all},  for GHZ states~\cite{kaniewski2016analytic,PhysRevLett.127.230503}, and for all multipartite entangled states that admit a Schmidt decomposition~\cite{vsupic2018self}\footnote{Notably, multipartite states do not generally allow for Schmidt decomposition, although some generalisations of the Schmidt decomposition have been explored in the literature~\cite{acin2000generalized,kraus2010local}.}. The idea of self-testing has been further extended to a prepare-and-measure scenario~\cite{tavakoli2020self, farkas2019self} as well as contextuality~\cite{bharti2019robust, bharti2019local}, and steering~\cite{vsupic2016self, gheorghiu2015robustness,shrotriya2020self} cases. For a comprehensive review of self-testing, see reference \cite{vsupic2020self}.

Self-testing with Bell nonlocality is an effective technique for learning about the inner workings of a pair of space-like separated entangled devices and, as such, has been translated to experiments \cite{Wang285,liu2018device,zhang2018experimentally,zhang2019experimental,Bancal2021selftestingfinite,PhysRevA.99.032108}. However, for single quantum devices, self-testing based on Bell nonlocality forfeits its relevance and one requires local self-testing schemes for such cases. Since computation typically happens locally, a quantum computer is a canonical example of a single quantum device. This necessitates the development of self-testing protocols for single untrusted quantum devices.

In this work, we take an approach towards self-testing via non-contextuality inequalities---linear inequalities \cite{PhysRevLett.101.020403,Cabello08}, similar to Bell inequalities, the experimental violation of which \cite{kirchmair2009state,lapkiewicz2011experimental} can be used to witness non-classicality. Their advantage is precisely that they can be tested on individual (rather than composite) systems and thus bypasses the non-communication assumption. Instead, as we shortly describe, we make a pragmatic assumption on the measurements that we call \emph{KCBS orthogonality}. Further, for an experimentally relevant self-testing protocol, the scheme must be \emph{robust} against experimental noise. In particular, obtaining maximal violations should not be a necessity for the scheme to work. %

We provide the first \emph{robustness curve} (plotted and explained in \Figref{experiment_result}, \textsc{Right}) which allows one to self-test a single quantum system, based on contextuality and experimentally realise this self-test.
To obtain the robustness curve, we utilise semidefinite programming with moment matrices. The experiment (see \Figref{experiment_result}, \textsc{Left}) uses a single trapped $^{40}{\rm Ca}^+$ ion and is the first to simultaneously address (see Table~\ref{tab:comparison}) the detection efficiency loophole, randomness of basis selection, and to quantify the deviation from the assumption which permits self-testing (closely related to compatibility and sharpness of measurements; see \Subsecref{RelationToPriorWork}). In this case, the assumption is KCBS orthogonality. We find that the experimental data is consistent with our self-testing robustness curve which lower bounds the total fidelity of the measurements and state---a measure of closeness of the experimental setup to the ideal internal configuration.

\subsection{Overview and organisation}
We first outline our theoretical contributions (\Secref{TheoreticalFramework}) and subsequently delineate 
the experiment (\Secref{Experiment}). 
A self-testing scheme treats a quantum device like a black-box whose inner working cannot be accessed beyond giving it classical inputs and receiving classical outputs (\Subsecref{Theory_Background}). Every non-trivial self-test must make additional assumptions. Informally, our assumption which we termed KCBS orthogonality (\Assuref{KCBSexclusivity}, \Subsecref{Theory_robustKCBS}) is that the measurements are projective and satisfy orthogonality corresponding to a cyclic graph (vertices represent measurements and edges relations of orthogonality). %
We denote our self-test by $\mathcal{I}$ (\Eqref{self-test}) which is a linear expression in the input/output probabilities associated with the device. $\mathcal{I}$ can be made $0$ by a quantum device but every classical device must yield $\mathcal{I} \ge Q - C > 0$, where $Q$ and $C$ are constants, independent of the device once the number of measurements are fixed. It is known \cite{bharti2019robust} that the states and measurements associated with the device, henceforth referred to as a configuration, are unique for any device which yields $\mathcal{I}=0$; the uniqueness is up to a global isometry.  Such a configuration is termed the \emph{Klyachko-Can-Binicioğlu-Shumovsky (KCBS) configuration}. Note that, using global isometry, one can change any state $\ket{\phi}$ to any other state $\ket{\xi}$. This might prompt one to think that assuming the freedom  up to global isometries may render the problem trivial. This is not true. If there are two configurations that attain  $\mathcal{I}=0$,  we require the same global isometry to map states as well as measurement settings from one configuration to another. If one were to use different global isometries for mapping states and measurements, the problem would have been definitely trivial.  Further, for a single device, the notion of local isometry does not make sense. Consider a configuration for which $\mathcal{I}$ is close to zero. To quantify the closeness of it to the KCBS configuration, we define \emph{total fidelity}, $F$ (\Eqref{fidelity}). $F$ is six for the KCBS configuration.\footnote{The number $6$ appears because the KCBS configuration is specified by five measurements and a state; the fidelity to each is summed to obtain total fidelity.} We give a lower bound on $F$ as a function of $\mathcal{I}$. The idea is (\Subsecref{ProofOverview}) to define an isometry in terms of the measurement operators of the device itself, which maps a device's configuration to the KCBS configuration when $\mathcal{I}=0$. The expression for $F$ when $\mathcal{I}>0$ becomes an optimisation problem which can be relaxed to a hierarchy of semidefinite programs (SDPs). This in turn can be solved numerically to obtain lower bounds on $F$. The details of the proof appear in \Subsecref{lowerBoundSDPrelaxation} and those of the numerical solution in \Subsecref{numerics}. Obtaining the numerical solution was challenging. This is primarily because the description of the SDP we obtained is implicit. Therefore, we first find this explicit description by performing symbolic computation and subsequently, solve the resulting SDP. %
After performing %
symbolic computation for finding the objective and the constraints, we obtain an SDP in $18,336$ variables (Hermitian matrix of size $192\times 192$) with $16,859$ constraints. 
The size of the matrix is determined by the number and type of terms in the objective function. We emphasise that even the relevant representation of the objective function itself (in terms of variables we optimise over) already turns out to be quite long and would have been difficult to obtain without symbolic computation. 

We apply this self-test to an experimental setup based on a single trapped \calcium ion. The KCBS configuration can be realised using a qutrit and projective measurements. In the experiment, the three levels of the qutrit are formed by the ground state, and the first excited state which is further split by the presence of a magnetic field. The measurements are executed by rotating between the ground state and the excited state using $729$ nm lasers, performing a photon number measurement   (via fluorescence detection), followed by an inverse of the rotation (see \Figref{experiment_result}, \textsc{Left} and \Subsecref{setup_outline}). The initial state and the rotation angles can be chosen such that they correspond directly with the KCBS configuration (\Subsecref{KCBSconfig}). The actual experimental setup, of course, requires numerous details to be first addressed (\Subsecref{setup_outline}) for this correspondence to work. Using this setup, we perform two types of experiments (\Subsecref{theExp}), both of which are deviations from the KCBS configuration, constructed to ensure \Assuref{KCBSexclusivity} (discussed above) holds. First, we leave the measurements unchanged but progressively depolarise the initial state. Second, we change the measurement settings, parametrised by an angle $\theta$ and tilt the initial state. In effect, these experiments determine $\mathcal{I}$ and the robustness curve immediately lower bounds the total fidelity of the underlying configuration to the KCBS configuration. By performing tomography, we determine the actual fidelity to the KCBS configuration and find that the lower bound is indeed satisfied, as expected (\Figref{experiment_result}, \textsc{Right}; \Subsecref{the_data}).

\subsection{Relation to prior work\label{subsec:RelationToPriorWork}}
The first local self-testing scheme was presented in reference \cite{bharti2019robust}, where the scheme relied on violation of non-contextuality inequalities via qutrits.
Self-testing schemes for single device has since been a topic of  active investigation. The approach based on contextuality was further explored in~\cite{saha2020sum} and extended to arbitrarily high-dimensional quantum systems~\cite{bharti2019local}. The primary limitation of these approaches is that they do not yield the complete robustness curve; in particular, while~\cite{bharti2019robust} construct a robust self-testing scheme, this robustness is only shown up to multiplicative factors and is thus not directly applicable to real-world experimental scenarios.

Subsequently, in a breakthrough result, a self-test for a single quantum system was introduced in~\cite{metger2020self} which relied on cryptographic assumptions, as opposed to contextuality. The advantage of this approach over contextuality is that it is in the fully device independent setting, i.e., where one only accesses the input/output behaviour of the device. %
On the other hand, all contextuality based approaches reported so far make some assumptions, however mild, about the internal working of the device.\footnote{These assumptions may be dropped but then the robustness curve only holds against ``memoryless model'', i.e., the guarantee holds in a very restricted setting.} Therefore, theoretically cryptographic self-tests are appealing. However, in practice they are harder to implement since cryptographic functions require coherent control of a large number of qubits (the number of qubits determine the security). %
Based on a recent works (e.g.~\cite{kahanamoku2021classically}), one can estimate that such cryptographic functions would require a circuit with at least $1000$ qubits and a depth of $10000$ to be implemented.
This limitation is circumvented by our contextuality based robust self-test.

We now compare our experiment to previous experiments related to contextuality. To this end, we note that for demonstration of contextuality, it is conventional to adopt an operational language. In this work, however, we take a more pragmatic approach and work directly under the assumption that quantum mechanics governs the behaviour of our devices. Consequently, the terminology we use to characterise our experiment is slightly different from prior works. More precisely, the analogue of our KCBS orthogonality requirement, in the operational perspective may be thought of as compatibility and sharpness of measurements (see \Subsecref{noiseComment}). In our experiment, we simultaneously quantify and minimise divergence from the basic assumptions---perfect detection, random selection of measurements, compatibility and sharpness---which, to the best of our knowledge, makes it the most comprehensive contextuality experiment to be reported (see Table~\ref{tab:comparison}).

The construction of swap isometry in our paper is motivated by~\cite{bancal2015physical}. While the construction  in~\cite{bancal2015physical} works for Bell scenarios, our construction applies for single party setting.

\begin{figure*}[htbp]
	\centering
	\subfigimg[width=0.45\textwidth]{}{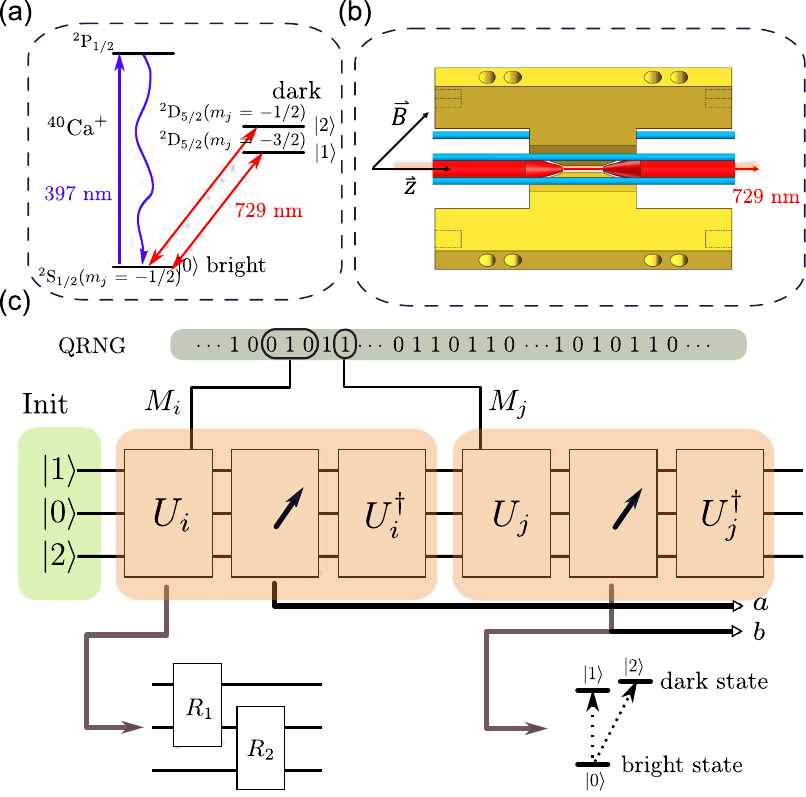}
	\hspace{0.5in}
	\subfigimg[width=0.45\textwidth]{}{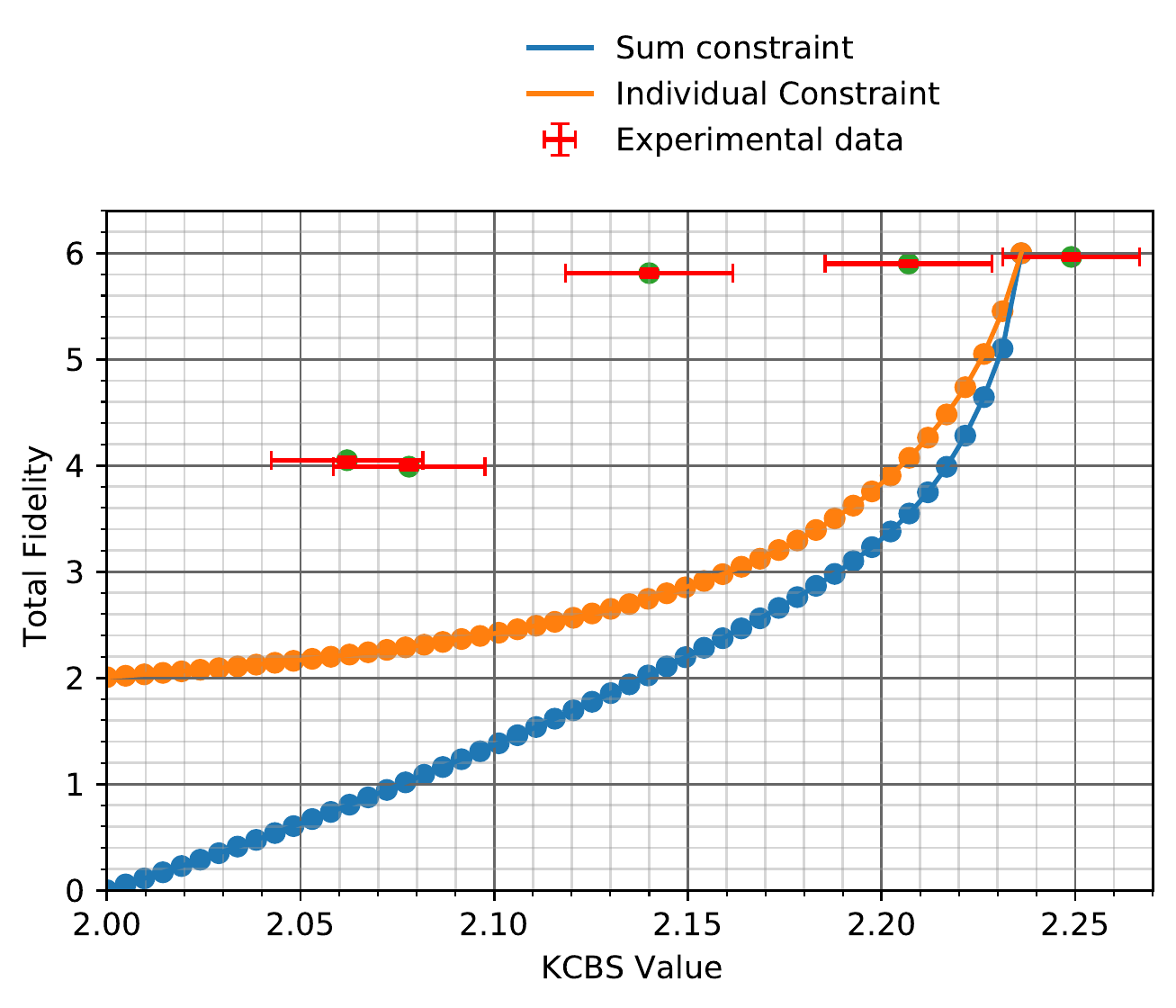}
	\caption{%
  \textsc{Left.} The $^{40}{\rm Ca}^{+}$ ion system and the measurement scheme. (a) The energy level diagram of $^{40}{\rm Ca}^{+}$ ion. (b) A schematic of the experimental setup. (c) The sequence of operations used in our experiment to measure $\Pr_i$ via $\Pr(10|ij)$, i.e., the probability of obtaining $(a,b)=(1,0)$ when $M_i$ and $M_j$ are measured. We choose $(i,j)\in E$ (edges of a five cycle graph) at random using a quantum random number generator. Each measurement $M_i$ is implemented by performing a rotation $U_i$, a fluorescence measurement (performed using a photo-multiplier tube), and undoing the rotation. For each measurement, the experiment is repeated over 10000 times.  \\
  \textsc{Right.} The $x$-axis represents $\sum_k p_k$ (which we call the KCBS value) while the $y$-axis represents the total fidelity, $F$ to the ideal KCBS configuration. Recall $\mathcal{I}=Q_n - \sum_k p_k$. The bottom curve (blue), obtained by solving a semidefinite program, lower bounds $F$ as a function of $\sum_k p_k$. The top curve (orange) is similarly obtained by assuming, in addition, that $p_1=p_2=\dots p_5$. The points represent various experiments which (nearly) satisfy \Assuref{KCBSexclusivity}. The $x$ coordinate of these points represents $\sum_k p_k$ while the $y$ coordinate is evaluated by performing tomography and computing the fidelity to the ideal KCBS configuration. We note that $\sum_k p_k$ exceeds $\sqrt{5}$ (the maximum quantum value). This is warranted because our experiments are nearly perfect, and even small deviations from the underlying assumptions cause the KCBS value to increase. It is worth noting that the right end of the error bar, however, is well within the $\sqrt{5}$ limit.} %
	\label{fig:experiment_result}
\end{figure*}

\begin{table*}
  \centering
  \begin{tabular}{ccccccc}
  \hline
  Year & System                                         & Self-test  & Compatibility & Sharp measurement & Random basis selection & Detection loophole                \\ \hline
  2011 & Photon~\cite{lapkiewicz2011experimental}        & $\times$  & quantified         & not checked          & $\times$                   & not addressed                               \\ 
  2013 & $\text{Yb}^{+}$ ion~\cite{um2013experimental}          & $\times$   & quantified          & not checked         & $\times$                   & closed                             \\ 
  2016 & Superconducting~\cite{jerger2016contextuality}  & $\times$  & quantified          & not checked         & $\times$                   & closed                              \\ 
  2017 & Photon~\cite{PhysRevLett.119.220403}            & $\times$  & quantified          & not checked         & $\times$                   & not addressed                                 \\ 
  2018 & $\text{Ca}^{+}$ ion~\cite{PhysRevA.98.050102}          & $\times$   & quantified          & not checked         & $\times$                   & closed                             \\ 
  2019 & Photon~\cite{PhysRevLett.122.080401}            & $\times$  & quantified          & not checked         & $\times$                   & not addressed                               \\ 
  2020 &  $\text{Ba}^{+}$ ion~\cite{PhysRevApplied.13.034077}   & $\times$   & quantified          & not checked         & $\times$                   & closed                             \\ 
  2022 & $\text{Ca}^+$ ion [this work]                                     & \checkmark  & quantified          & quantified            & \checkmark               & closed                              \\ \hline
  \end{tabular}
  \caption{A comparison of our experiment to previous experimental results on contextuality (see \Subsecref{noiseComment}).} %
  \label{tab:comparison}
  \end{table*}

\section{Theoretical framework \label{sec:TheoreticalFramework}}
\global\long\def\KCBS{\text{KCBS}}%
\global\long\def\tr{\text{tr}}
\global\long\def\swap{{\rm SWAP}}%

\subsection{Background \label{subsec:Theory_Background}}

Suppose we have a set of (independent) identical devices each of which
takes a string as input and produces a string as output. Informally,
we say the device is \emph{self-testing} if there exists some constraints
on the input-output behaviour of the device, which if satisfied, guarantees
the integrity of the device. More precisely, let $p_{i}$ denote the
probability of event $i$, where an event is a set of input output
behaviours. Let $\mathcal{I}:=\sum_{i}q_{i}p_{i}$ denote a linear
combination of these probabilities, where $q_{i}\in\mathbb{R}$. We
say that $(\mathcal{I},c)$, for some $c\in\mathbb{R}$, self-tests
the device if the following holds---if the device satisfies $\mathcal{I}=c$,
then the device is uniquely described as containing a specific quantum
state, the inputs corresponding to selecting specific measurements,
the outputs correspond to the action of the measurement on the aforementioned
quantum state, up to a global isometry. Clearly, no such self-test
can exist unless one makes additional assumptions.\footnote{This is simply because the device could classically store and reproduce
whatever statistics the self-test is required to satisfy.} Three assumptions have been explored in the literature---spatial
separation, computational assumptions and orthogonality of measurements.\footnote{First, physical separation between the components of the device which
may be termed the ``Bell setting''. A self-test in this setup harnesses
nonlocal correlations which can be produced using quantum devices
but not by classical devices. Second, computational hardness assumptions
(such as Learning With Errors), i.e., certain tasks are assumed hard
for the device being self-tested. A self-test in this setup harnesses
the ability of a quantum device to evaluate cryptographic functions
in superpositions and produce correlations which a classical device
cannot.} In this work, we focus on the latter which requires that the measurements
are repeatable and satisfy an orthogonality relation among them---for
instance if one obtains the outcome corresponding to $\Pi_{0}$, then
the outcome corresponding to $\Pi_{1}$ cannot occur (upon subsequent
measurement). This assumption may be motivated physically by considering
a form of tomography where the experimentalist is only required to
ensure that the measurements performed in the device are reliable
and by construction satisfy the orthogonality condition.\footnote{A self-test in this setup harnesses the ability of quantum devices
to produce ``non-contextual'' correlations. Bell nonlocality may
be seen as a special case of contextuality.}

\subsection{Robust KCBS self-testing \label{subsec:Theory_robustKCBS}}

We make the following assumption about the device we wish to self-test.
To describe it, we setup some notations. We denote the $i$th \emph{binary measurement} by $M_{i}:=(\Pi_{0|i},\Pi_{1|i})$, where $\Pi_{0|i}$
denotes the measurement operator corresponding to the zeroth output
and $\Pi_{1|i}$ denotes that corresponding to the first output. Being
measurement operators, they satisfy the probability conservation condition
$\Pi_{0|i}^{\dagger}\Pi_{0|i}+\Pi_{1|i}^{\dagger}\Pi_{1|i}=\mathbb{I}$. We say the binary measurement is \emph{repeatable} and \emph{Hermitian} if $\Pi_{b|i}^2= \Pi_{b|i}$ for $b\in{0,1}$ and $\Pi_{b|i}^{\dagger}=\Pi_{b|i}$, respectively.

\begin{assumption}
\label{assu:KCBSexclusivity}For an odd $n$, the quantum state in
the device is $\rho$ and the measurements $M_{1},M_{2},\ldots, M_{n}$
are repeatable and Hermitian binary measurements (as described above) satisfying $\Pi_{1|i}\cdot\Pi_{1|j}=0$
for $\{i,j\}\in\{\{1,2\},\{2,3\},\ldots,\{n-1,n\},\{n,1\}\}=:E$.
\end{assumption}

\begin{rem}
While one could use Naimark's theorem (according to which the statistics of any quantum measurement can also be obtained using a projective measurement, since any a positive-operator-valued measure can be seen as a projective-valued measure in a larger Hilbert space) to argue that restricting to projective measurements, this conclusion no longer holds for sequential measurements where post-measurement states become relevant. Further, Hermiticity too cannot be deduced by this argument. However, these requirements are not too strong and our experimental setup nearly satisfies the aforementioned assumption (see \Subsecref{noiseComment}).
\end{rem}

For a device satisfying \Assuref{KCBSexclusivity}, we consider the
self-test\footnote{This particular form is of interest because, informally, any classical
(i.e., non-contextual) model satisfying \Assuref{KCBSexclusivity}
cannot result in $\sum_{i=1}^{n}p_{i}>C_{n}$, where $C_{n}:=(n-1)/2$
while quantumly, for the KCBS configuration (see \Defref{idealKCBSconfig})
one obtains $\sum_{i=1}^{n}p_{i}=Q_{n}$, the highest possible~\cite{AQBTC13}. The self-test measures how close one is to the maximum
quantum value, $Q_{n}$.} $\mathcal{I}$
\begin{equation}
\mathcal{I}:=Q_{n}-\sum_{i=1}^{n}p_{i},
\label{eq:self-test}\end{equation}
where $p_{i}:=\tr(\Pi_{i}\rho)=\left\langle \Pi_{i}\right\rangle $
and $Q_{n}:=\frac{n\cos(\pi/n)}{1+\cos(\pi/n)}$ is the ``quantum
value''.\footnote{The maximum value of $\sum_i p_i$ that is \emph{classically} achievable (i.e., when all measurements commute) is $(n-1)/2$.} This quantum value is achieved by $\sum_{i}\left\langle \Pi_{i}\right\rangle $
for the \emph{KCBS configuration} ($\rho^{\KCBS},\{\Pi_{i}^{\KCBS}\})$,
where $\rho^{\KCBS}$ is a qutrit in a pure state and $\Pi_{i}^{\KCBS}$
are projectors (see \Defref{idealKCBSconfig} below), yielding $\mathcal{I}=0$.
It was shown in \cite{bharti2019robust} that any \emph{quantum realisation}
$(\rho,\{\Pi_{i}\})$ which satisfies \Assuref{KCBSexclusivity} and
yields $\mathcal{I}=0$, must be the same as the KCBS configuration,
up to a global isometry.\footnote{In fact, they also showed that if $\mathcal{I}=\epsilon$ is close
to zero, then the quantum realisation must be close to the KCBS configuration.
However, their analysis did not yield the exact function corresponding to~\eqref{fidelity}---they only
had an asymptotic bound (on the fidelity to the KCBS configuration) of the form $\mathcal{O}(\sqrt{\epsilon})$.}

For clarity, we restrict to $n=5$ henceforth but the arguments readily
generalise. Suppose the device corresponds to a quantum realisation
$(\rho,\{\Pi_{i}\})$ for which $\mathcal{I}=\epsilon$ is small but
not exactly zero. In this case, we derive a lower bound on the fidelity
of $(\rho,\{\Pi_{i}\})$ to $(\rho^{\KCBS},\{\Pi_{i}^{\KCBS}\})$
(up to a global isometry). More precisely, we lower bound the value
of the following function
\begin{equation}
F:=\min_{\rho,\{\Pi_{i}\}}\max_{V}\left[\sum_{i=1}^{5}\mathcal{F}(\tr_{\mathcal{H}}[V\Pi_{i}\rho\Pi_{i}V^{\dagger}],\Pi_{i}^{\KCBS}\rho^{\KCBS}\Pi_{i}^{\KCBS})+\mathcal{F}(\tr_{\mathcal{H}}[V\rho V^{\dagger}],\rho^{\KCBS})\right],\label{eq:fidelity}
\end{equation}
where $(\rho,\{\Pi_{i}\}_{i=1}^{5})$ is a quantum realisation %
of $\{p_{i}\}_{i=1}^{5}$ satisfying $\mathcal{I}=\epsilon$, $V$ is an isometry from $\mathcal{H}$
to $\mathcal{H}\otimes\mathcal{H}^{\KCBS}$, i.e., from the space on
which $\rho,\{\Pi_{i}\}_{i=1}^{5}$ are defined, to itself tensored
with the $3$-dimensional space where $\rho^{\KCBS},\{\Pi_{i}^{\KCBS}\}_{i=1}^{5}$
are defined, and $\mathcal{F}(\sigma,\tau):={\rm tr}\sqrt{\left|\sigma^{1/2}\tau^{1/2}\right|}$
is the fidelity of $\rho\ge0$ to $\sigma\ge0$. This lower bound can
be approximated by a sequence of semi-definite programs and can
be evaluated for any\footnote{Instead of the sum, $\sum_i p_i$, one could impose constraints on the values of $p_i$ individually. In fact, this is what we do in our analysis as it yields a better bound.} %
$\mathcal{I}=\epsilon$.

\subsection{Proof overview\label{subsec:ProofOverview}}
We briefly describe an algorithm for estimating $F$ and outline an argument which shows that this algorithm always yields a lower bound (on $F$).
First, we drop the maximisation
over $V$ in \Eqref{fidelity} and replace it with a particular isometry
$V$ which is expressed in terms of $\rho,\{\Pi_{i}\}_{i=1}^{5}$. %
Then, as we shall see, the expression for the fidelity appears as
a sum of terms of the following form. Let $w$ be a \emph{word} created
from the \emph{letters} $\{\mathbb{I},\Pi_{1},\Pi_{2}\dots\Pi_{5},\hat{P}\}$
with $\hat{P}^{\dagger}\hat{P}=\mathbb{I}$, $\Pi_{i}^{2}=\Pi_{i}$
and $\Pi_{i}\Pi_{j}=0$ if $(i,j)\in E$, i.e., when $i,j$ are
exclusive.\footnote{We introduced $\hat{P}$ for completeness; its role is explained later.}
The fidelity is a linear combination of these words, i.e., $F=\min_{\left\{ \left\langle w\right\rangle \right\} }\sum_{w}\alpha_{w}\left\langle w\right\rangle$,
where $\tr[w\rho]=:\left\langle w\right\rangle $, subject to the
constraint that $\{\left\langle w\right\rangle \}_{w}$ corresponds
to a quantum realisation. The advantage of casting the problem in
this form is that one can now relax the problem to a sequence of semi-definite
programs.\footnote{I.e., construct an NPA-like hierarchy~\cite{NPA}.}
The idea is simple to state. Treat $\left\{ \left\langle w\right\rangle \right\} _{w}$
as a vector. Denote by $Q$ the set of all such vectors which correspond
to a quantum realisation (of $\left\{ p_{i}\right\} _{i=1}^{5}$).
It turns out that one can impose constraints on words with $k$ letters,
for instance. Under these constraints, denote by $Q_{k}$ the set
that is obtained. Note that $Q_{k}\supseteq Q$ for it may contain
vectors which don't correspond to the quantum realisation. In fact,
$Q_{k}$ can be characterised using semi-definite programming constraints
(which in turn means they are efficiently computable). Intuitively,
one expects that $\lim_{k\to\infty}Q_{k}=Q$. Further, it is clear
that $F=\min_{\{\left\langle w\right\rangle \}_{w}\in Q}\sum_{w}\alpha_{w}\left\langle w\right\rangle \ge\min_{\left\{ \left\langle w\right\rangle \right\} _{w}\in Q_{k}}\sum_{w}\alpha_{w}\left\langle w\right\rangle $
as we are minimising over a larger set on the right hand side.

\subsection{Proof | Lower bound using an SDP relaxation \label{subsec:lowerBoundSDPrelaxation}}

We begin with defining the ideal KCBS configuration, referred to above.
\begin{defn}[An ideal KCBS configuration]
\label{def:idealKCBSconfig}Consider a three-dimensional Hilbert
space spanned by the basis $\{\left|0\right\rangle ,\left|1\right\rangle ,\left|2\right\rangle \}$.
Let
\begin{equation}
\left|u_{l}\right\rangle :=\cos\theta\left|0\right\rangle +\sin\theta\sin\phi_{l}\left|1\right\rangle +\sin\theta\cos\phi_{l}\left|2\right\rangle,
\end{equation}
 where $\phi_{l}:=l\pi(n-1)/n$ for $1\le l\le n$ and $\cos^2{\theta} = \frac{\cos{(\pi/n)}}{1+ \cos{(\pi/n)}}$. Define
\begin{align*}
\left|\psi^{\KCBS}\right\rangle  & :=\left|0\right\rangle, \\
\Pi_{i}^{\KCBS} & :=\left|u_{i}\right\rangle \left\langle u_{i}\right|.
\end{align*}
\end{defn}
\global\long\def\swap{{\rm SWAP}} 
For concreteness, we first consider a unitary $U_{\swap}$ instead
of an isometry, which acts on two spaces $\mathcal{A}$ and $\mathcal{A}'$,
i.e., $U_{\swap}:\mathcal{A}\otimes\mathcal{A}'\to\mathcal{A}\otimes\mathcal{A}'$. 
The space $\mathcal{A'}$ is three-dimensional and $\mathcal{A}$ is an arbitrary Hilbert space. Informally, we want to construct the unitary $U_{\swap}$ to be such that it takes a realisation in $\mathcal{A}$ and maps it to a realisation in $\mathcal{A}'$ which has a large overlap with the KCBS configuration.
To this end, we first assume that $\mathcal{A}$ is three-dimensional. In particular, suppose that the $\mathcal{A}$ register is in the
state
\begin{equation}
\sigma\in\underbrace{\{\left|\psi^{\KCBS}\right\rangle \left\langle \psi^{\KCBS}\right|\}\cup
\{ \Pi_{i}^{\KCBS} \ket{\psi^{\KCBS}} \bra{\psi^{\KCBS}} \Pi_i^{\KCBS}  \}   _{i=1}^{5}}_{:=\mathcal{S}^{\KCBS}}.
\end{equation}
Then, at the very least, we want $U_{\swap}$ to map $\sigma_{\mathcal{A}}\otimes\left|0\right\rangle \left\langle 0\right|_{\mathcal{A}'}$
to $\left|0\right\rangle \left\langle 0\right|_{\mathcal{A}}\otimes\sigma_{\mathcal{A}'}$
(see \Figref{IllustrationSWAP}). Our strategy is to construct $U_{\swap}$
in this seemingly trivial case and then express it in terms of the
state and measurement operators on $\mathcal{A}$. The rationale is
that by construction, $U_{\swap}$ will work for the ideal case and
therefore should also work for cases close to ideal. This should become
clear momentarily. Note that any circuit that swaps two qutrits should
let us achieve our simplified goal (because all elements of $S^{\KCBS}$
are defined on a three dimensional Hilbert space). One possible qutrit
swapping unitary/circuit (a special case of the general qudit swapping
unitary/circuit defined in~\cite{bancal2015physical})
may be defined
as $S_{\swap}'':=TUVU$, where
\begin{align}
T & :=\mathbb{I}_{\mathcal{A}}\otimes\sum_{k=0}^{2}\left|-k\right\rangle \left\langle k\right|_{\mathcal{A}'},\nonumber \\
U & :=\sum_{k=0}^{2}P_{\mathcal{A}}^{k}\otimes\left|k\right\rangle \left\langle k\right|_{\mathcal{A}'},\nonumber \\
V & :=\sum_{k=0}^{2}\left|\bar{k}\right\rangle \left\langle \bar{k}\right|_{\mathcal{A}}\otimes P_{\mathcal{A}'}^{-k},\label{eq:TUV}
\end{align}
where $P:=\sum_{i=0}^{2}\left|\overline{k+1}\right\rangle \left\langle \bar{k}\right|$
is a translation operator, the arithmetic operations are modulo 3 and $\{\left|\bar{0}\right\rangle ,\left|\bar{1}\right\rangle ,\left|\bar{2}\right\rangle \}$
is a basis for the qutrit space $\mathcal{A}$. %
We omit the proof that $S''_{\swap}$ indeed performs a swap operation (for a proof see~\cite{bancal2015physical}). To generalise this idea and to construct an isometry, we relax the assumption that
$\mathcal{A}$ is a three dimensional Hilbert space. We re-express/replace
the operations in $T,U,V$ which act on the $\mathcal{A}$ space by
linear combinations of monomials in $\{\Pi_{i}\}_{i=1}^{5}$. We obtain
the coefficients used in these linear combinations by assuming the
space $\mathcal{A}'$ is three dimensional. The idea is simply that this map reduces
to a swap operation when we re-impose the assumptions and for cases
close to it, we expect it to behave appropriately. %
We describe
this procedure more precisely below.
\global\long\def\cross{\text{cross}}%

\begin{lyxalgorithm}[Constructing an isometry]
 Let\label{alg:ConstructingAnIsometry}
\begin{itemize}
\item $\mathcal{A}'$ be a three-dimensional Hilbert space spanned by an
orthonormal basis $\left\{ \left|0\right\rangle _{\mathcal{A}'},\left|1\right\rangle _{\mathcal{A}'},\left|2\right\rangle _{\mathcal{A}'}\right\} $,
\item $\rho^{\KCBS},\{\Pi_{i}^{\KCBS}\}_{i=1}^{5}$ be an ideal KCBS configuration
(see \Defref{idealKCBSconfig}) on $\mathcal{A}'$,
\item $\mathcal{A}$ be a Hilbert space with dimension at least $3$ containing
orthonormal vectors $\left\{ \left|\bar{0}\right\rangle _{\mathcal{A}},\left|\bar{1}\right\rangle _{\mathcal{A}},\left|\bar{2}\right\rangle _{\mathcal{A}}\right\}$,
\item $\rho,\{\Pi_{i}\}_{i=1}^{5}$ be an arbitrary quantum realisation defined on $\mathcal{A}$.
\end{itemize}
Define $T,U,V$ as in \Eqref{TUV} and let $S'_{\swap}:=TUVU$ with
the following changes. Let
\begin{equation}
\mathcal{W}^{\KCBS}:=\{\{\Pi_{i}^{\KCBS}\}_{i=1}^{5},\{\Pi_{i}^{\KCBS}\Pi_{j}^{\KCBS}\}_{i,j=1}^{5},\dots\}
\end{equation}
 and
\begin{equation}
\mathcal{W}:=\{\{\Pi_{i}\}_{i=1}^{5},\{\Pi_{i}\Pi_{j}\}_{i,j=1}^{5}\dots\}
\end{equation}
 be the set of ``words'' formed by the KCBS projectors and those
of the arbitrary quantum realisation, respectively.
\begin{enumerate}
\item Translation operator:
\begin{enumerate}
\item Express $P_{\mathcal{A}'}$ as a linear combination of elements in
$\mathcal{W}$, i.e., $P_{\mathcal{A}'}=\sum_{l^{\KCBS}\in\mathcal{W}^{\KCBS}}\alpha_{l}l^{\KCBS}$.
\item Define $P_{\mathcal{A}}:=\sum_{l\in\mathcal{W}}\alpha_{l}l$.
\end{enumerate}
\item Basis projectors: Formally replace, in $V$, the operators
\begin{enumerate}
\item $\left|\bar{0}\right\rangle \left\langle \bar{0}\right|_{\mathcal{A}}$
by $\Pi_{1}$,
\item $\left|\bar{1}\right\rangle \left\langle \bar{1}\right|_{\mathcal{A}}$
by $\Pi_{2}$, and
\item $\left|\bar{2}\right\rangle \left\langle \bar{2}\right|_{\mathcal{A}}$
by $(\mathbb{I}-\Pi_{1})(\mathbb{I}-\Pi_{2})$.
\end{enumerate}
Thus, $V$ now becomes $\Pi_{1}\otimes\mathbb{I}_{\mathcal{A}'}+\Pi_{2}\otimes P_{\mathcal{A}'}^{-1}+(\mathbb{I}-\Pi_{1})(\mathbb{I}-\Pi_{2})\otimes P_{\mathcal{A}'}^{-2}$.
\end{enumerate}
\end{lyxalgorithm}

We found an explicit linear combination for step 1 (a) of \Algref{ConstructingAnIsometry}.\footnote{Using prior results, we expect one can prove the existence of such
a linear combination, but we do not pursue this here.} Thus, the algorithm always succeeds at constructing $S'_{\swap}$.
We must show that $S'_{\swap}$ is in fact an isometry. This is important
because of the following reason. Recall that our objective was to
lower bound \Eqref{fidelity}. To this end, we said we drop the maximisation
over all possible $V$s, (for a given quantum realisation $(\rho,\left\{ \Pi_{i}\right\} _{i=1}^{5})$)
and instead insert a specific isometry $S'_{\swap}$ (which is a function
of $(\rho,\{\Pi_{i}\}_{i=1}^{5})$). Note that this argument for lower
bounding \Eqref{fidelity} breaks if $S'_{\swap}$ is not an isometry.
In fact, $P_{\mathcal{A}}$ as produced by the algorithm is not necessarily
unitary (viz. $P_{\mathcal{A}}^{\dagger}P_{\mathcal{A}}=\mathbb{I}_{\mathcal{A}}$
may not hold). To address this, we use the \emph{localising matrix}
technique introduced in ~\cite{pironio2010convergent}. Let $\hat{P}_{\mathcal{A}}$ be
unitary matrix satisfying $\hat{P}P\ge0$, where we dropped the subscript
for clarity. Consider the case where $P$ is not unitary. In that
case, one can use polar decomposition to write $P=\left|P\right|U$
(not to be confused with the $U$ above; where $\left|P\right|\ge0$
and $U^{\dagger}U=\mathbb{I}$) so choosing $\hat{P}=U^{\dagger}$
satisfies $\hat{P}P\ge0$. Thus for each $P$, the constraint can
be satisfied. Consider the other case, i.e., where $P=U$ is unitary.
Then,\footnote{%
To see this, write the
polar decomposition of $\hat{P}=\left|\hat{P}\right|E$ (where $\left|\hat{P}\right|\ge0$
and $E^{\dagger}E=\mathbb{I}$). Then, $\hat{P}U=\left|\hat{P}\right|EU$,
which is a polar decomposition of $\hat{P}U$. The polar decomposition
of a positive semi-definite matrix $M$ always has the form $M.\mathbb{I}$.
Using $M=\hat{P}U$, and identifying $EU$ with $\mathbb{I}$, we
have $E=U^{\dagger}$.} $\hat{P}=U^{\dagger}$. Thus, in the ideal case, we recover the same
unitary and for the case close to ideal, we are guaranteed that there
is some solution (which we expect should also work reasonably). Combining
these, we can construct $S_{\swap}$, which is an isometry.
\begin{lem}[$S_{\swap}$ is indeed an isometry]
 Let $S'_{\swap}$ be the map produced by \Algref{ConstructingAnIsometry}
and define $S_{\swap}$ to be $S'_{\swap}$ with $P_{\mathcal{A}}$
replaced by $\hat{P}_{\mathcal{A}}$. %
Then $S_{\swap}$ is an isometry if
the following conditions hold \label{lem:UswapIsAnIsometry}
\begin{align}
\hat{P}_{\mathcal{A}}^{\dagger}\hat{P}_{\mathcal{A}} & =\mathbb{I}_{\mathcal{A}},\\
\hat{P}_{\mathcal{A}}P_{\mathcal{A}} & \ge0.
\end{align}
\end{lem}

\begin{proof}
It suffices to show that $\left\langle \psi\right|_{\mathcal{A}'\mathcal{A}}S_{\swap}^{\dagger}S_{\swap}\left|\psi\right\rangle _{\mathcal{A}'\mathcal{A}}=1$
for all normalised $\left|\psi\right\rangle _{\mathcal{A}'\mathcal{A}}$.
We express $S_{\swap}=TUVU$, where $T$ is as in \Eqref{TUV}, $U:=\sum_{k=0}^{2}\hat{P}_{\mathcal{A}}^{k}\otimes\left|\bar{k}\right\rangle \left\langle \bar{k}\right|_{\mathcal{A}'}$
and $V:=\Pi_{1}\otimes\mathbb{I}_{\mathcal{A}'}+\Pi_{2}\otimes P_{\mathcal{A}'}^{-1}+(\mathbb{I}-\Pi_{1})(\mathbb{I}-\Pi_{2})\otimes P_{\mathcal{A}'}^{-2}$.
Observe that $T^{\dagger}T=\mathbb{I}_{\mathcal{A}\mathcal{A}'}=U^{\dagger}U$
since $\hat{P}_{\mathcal{A}}^{\dagger}\hat{P}_{\mathcal{A}}=\mathbb{I}_{\mathcal{A}}$.
Further, we have that
\begin{align}
V^{\dagger}V & =\Pi_{1}\otimes\mathbb{I}_{\mathcal{A}'}+\Pi_{2}\otimes\mathbb{I}_{\mathcal{A}'}+(\mathbb{I}-\Pi_{1})(\mathbb{I}-\Pi_{2})\otimes\mathbb{I}_{\mathcal{A}'} & \because\Pi_{1}\Pi_{2}=0,P_{\mathcal{A}'}^{\dagger}P_{\mathcal{A}'}=\mathbb{I}_{\mathcal{A}'}\\
 & =\mathbb{I}_{\mathcal{A}\mathcal{A}'}.
\end{align}
Hence, $S_{\swap}^{\dagger}S_{\swap}=\mathbb{I}_{\mathcal{A}\mathcal{A}'}$
establishing that $S_{\swap}$ is in fact unitary and thus also an
isometry.
\end{proof}
\begin{figure}
\begin{centering}
\includegraphics[width=0.8\paperwidth]{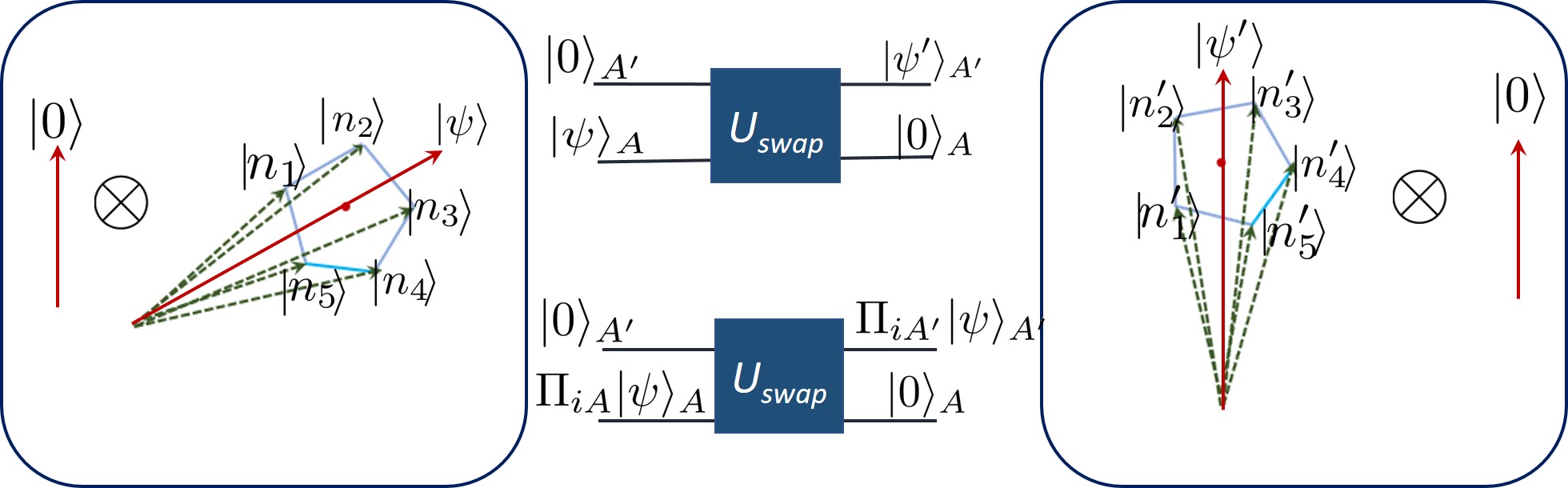}
\par\end{centering}
\caption{Illustration of how $U_{\text{SWAP}}$ acts on a quantum realisation in three dimensions. For the general case, this is expressed in terms of the measurement operators specifying the quantum realisation itself. It serves as a key ingredient in obtaining a bound on the total fidelity. %
\label{fig:IllustrationSWAP} }
\end{figure}
We can now combine all the pieces to write the final optimisation
problem we solve. We use $\mathcal{G}(\{a_{i}\})$ to denote the set
of all letters
\begin{lyxalgorithm}[The SDP for lower bounding the fidelity of a given realisation with
the ideal KCBS realisation]
 The algorithm proceeds in two parts.\label{alg:The-algorithm} \\
\textbf{Notation:} Let
\begin{itemize}
\item $\mathcal{A}$ represent an arbitrary %
Hilbert space and $\mathcal{A}'$
represent a three dimensional Hilbert space,
\item $\bar{\mathcal{W}}$ be the set of ``words'' constructed using $\{\Pi_{i}\}_{i=1}^{5}$,
$\hat{P}$ and $\hat{P}^{\dagger}$, i.e., $\bar{\mathcal{W}}=\mathcal{G}(\{\Pi_{i}\}_{i=1}^{5},\hat{P},\hat{P}^{\dagger})$,
and
\item define $\left\langle w\right\rangle :=\tr_{\mathcal{A}}(\rho w)$
and assume\footnote{Without loss of generality, one can purify a state by increasing the
dimension and have the operators leave that space unaffected (for
any given quantum realisation).} $\rho=\left|\psi\right\rangle \left\langle \psi\right|$ so that
$\left\langle w\right\rangle =\left\langle \psi\right|w\left|\psi\right\rangle $
for all $w\in\mathcal{W}$.
\end{itemize}
\textbf{Input:}
\begin{itemize}
\item (Implicit) $\Pi_{i}.\Pi_{j}=0$ for all $i,j\in E(G)$, where $G$
is a five cycle graph and $E$ are its edges, indexed $[1,2,\dots5]$.
\item Observed statistics:\footnote{To obtain a better bound, one can add more fine-grained statistics
as well such as the values of $p_{i}$ individually for each $i$.} $\sum_{i}p_{i}=\sum_{i}\left\langle \Pi_{i}\right\rangle =c$.
\end{itemize}
\textbf{Evaluation | Part 1.}
\begin{itemize}
\item Evaluate $S_{\swap}$ as described in \Lemref{UswapIsAnIsometry}
using \Algref{ConstructingAnIsometry}.
\item Define the objective function $f$ as
\begin{align}
f\left(\rho,\left\{ \Pi_{i}\right\} _{i=1}^{5}\right) & =\sum_{i=1}^{5}\mathcal{F}\left(\tr_{\mathcal{A}}\left[S_{\swap}\left(\Pi_{i}\rho_{\mathcal{A}}\Pi_{i}\otimes\left|0\right\rangle \left\langle 0\right|_{\mathcal{A}'}\right)S_{\swap}^{\dagger}\right],\Pi_{i}^{\KCBS}\rho_{\mathcal{A}'}^{\KCBS}\Pi_{i}^{\KCBS}\right)\\
 & \quad+\mathcal{F}\left(\tr_{\mathcal{A}}\left[S_{\swap}\left(\rho_{\mathcal{A}}\otimes\left|0\right\rangle \left\langle 0\right|_{\mathcal{A}'}\right)S_{\swap}^{\dagger}\right],\rho_{\mathcal{A}'}^{\KCBS}\right)
\end{align}
and evaluate the coefficients $f_{w}$ so that\footnote{This was done using symbolic computation.}
\begin{equation}
f=\sum_{w\in\mathcal{W}}f_{w}\left\langle w\right\rangle .
\end{equation}
\item Again, from \Algref{ConstructingAnIsometry}, evaluate $P_{\mathcal{A}}$
as a linear combination of $\left\langle w\right\rangle $.
\end{itemize}
\textbf{Evaluation | Part 2.} Solve the following SDP:
\begin{align*}
F_{\text{noiseless}}= & \min_{\{\left\langle w\right\rangle \}_{w\in\mathcal{W}}}\sum_{w\in\mathcal{W}}f_{w}\left\langle w\right\rangle \\
{\rm s.t.}\quad & \Gamma^{(k)}(\mathbb{I})\ge0 & \because\text{ all gram matrices are }\ge0\\
 & \Gamma^{(k)}(\mathbb{I})_{v,w}=\Gamma^{(k)}(\mathbb{I})_{v',w'} & \text{if }v^{\dagger}w=v^{\prime\dagger}w'\\
 & \Gamma^{(k)}(\hat{P}_{\mathcal{A}}P_{\mathcal{A}})\ge0 & \text{(localising matrix)}\\
 & \Gamma^{(k)}(\hat{P}_{\mathcal{A}}P_{\mathcal{A}})_{v,w}=\Gamma^{(k)}(\hat{P}_{\mathcal{A}}P_{\mathcal{A}})_{v',w'} & \text{if }v^{\dagger}\hat{P}_{\mathcal{A}}P_{\mathcal{A}}w=v^{\prime\dagger}\hat{P}_{\mathcal{A}}P_{\mathcal{A}}w'\\
 & \sum_{i=1}^{5}\left\langle \Pi_{i}\right\rangle =c & \text{(observed statistic)}
\end{align*}
where $\Gamma^{(k)}(X)$ is a matrix which is
\begin{itemize}
\item indexed by letters $w$,
\item whose matrix elements are given by $\Gamma^{(k)}(X)_{v,w}=\left\langle \psi\right|v^{\dagger}Xw\left|\psi\right\rangle $,
\item where $k$ defines the maximum number of letters that appear in the
words $w$ which index the matrix $\Gamma^{(k)}$, and,
\end{itemize}
where in the first two equality constraints, we use the following
relations:
\begin{itemize}
\item $\Pi_{i}.\Pi_{j}=0$ for all $i,j\in E(G)$ and
\item $\hat{P}_{\mathcal{A}}^{\dagger}\hat{P}_{\mathcal{A}}=\hat{P}_{\mathcal{A}}\hat{P}_{\mathcal{A}}^{\dagger}=\mathbb{I}_{\mathcal{A}}$.
\end{itemize}
\textbf{Output:} $F_{\text{noiseless}}(c)$.
\end{lyxalgorithm}

\begin{rem}
Note that $\hat{P}P\ge0\implies\Gamma^{(k)}(\hat{P}P)\ge0$. This
follows readily by letting $A^{\dagger}A=\hat{P}P$ for some $A$
(which must exist for any positive semi-definite matrix; one can use
spectral decomposition). Then $\Gamma^{(k)}(A^{\dagger}A)$ is a Gram
matrix and thus $\ge0$.
\end{rem}

\begin{rem}
We solved this SDP with $\left\langle w\right\rangle$ restricted to being real, but it has been shown that this is without loss of generality (see, e.g., §V.A in~\cite{bancal2015physical}). %
\end{rem}

We thus have an algorithm which can calculate the required lower bound
on fidelity, given the observed value of the KCBS operator.

\subsection{Numerics\label{subsec:numerics}}
As previously stated, obtaining the numerical solution was quite involved, owing to the fact that the description of the semidefinite program we obtain is implicit. Consequently, one  first needs to find this explicit description through symbolic computation and then solve the resulting semidefinite program. After performing the symbolic computation to find the objective and constraints, we solve an SDP over a $192\times192$ matrix with $16859$ constraints. The terms in the objective function must appear in the SDP matrix in order for us to impose appropriate constraints. This is because the optimisation program would otherwise be undefined. The unusually large size of the SDP matrices in our case reflects the complexity of the objective function in our case, which would have been difficult to obtain without symbolic computation. %
To capture all the terms that appear in the objective, we need words with three letters. Further, we require $769$ localising matrix constraints. 

We used Jupyter notebooks (based on python) as our programming environment. We used sympy \cite{sympy} for symbolic computations and cvxpy \cite{diamond2016cvxpy,agrawal2018rewriting} for translating the resulting SDP into an instance which MOSEK (an SDP solver) could solve. Our implementation is available on GitHub \cite{swapKCBS}.

\section{Experiment \label{sec:Experiment}}
We construct a physical setup which nearly satisfies \Assuref{KCBSexclusivity} and obtain the value of the self-test $\mathcal{I}$. Using the robustness curve (see \Figref{experiment_result}), we obtain a lower bound on the fidelity of our quantum realisation to the ideal KCBS configuration. Independently, we perform tomography on our device and again evaluate the fidelity of our quantum realisation to the ideal KCBS configuration. We find that the fidelity obtained by direct tomography is indeed lower bounded by the fidelity predicted by the robustness curve.

\subsection{The setup | Outline \label{subsec:setup_outline}}

We describe our experimental setup which is designed to realise the ideal KCBS configuration. We use a $^{40}{\rm Ca}^{+}$ ion trapped in a blade-shaped Paul trap. The qutrit basis states are encoded into three Zeeman sub-levels of the $^{40}{\rm Ca}^{+}$ ion, with
\begin{align}
  \left|0\right\rangle &= \left|S_{1/2},m_{J}=-1/2\right\rangle,\\
  \left|1\right\rangle &= \left|D_{5/2},m_{J}=-3/2\right\rangle,\\
  \left|2\right\rangle &= \left|D_{5/2},m_{J}=-1/2\right\rangle,
\end{align}
as shown in \Figref{experiment_result}a.
Unitary operations on these qutrits are performed by shining a linear polarized narrow-linewidth 729 $\rm nm$ laser beam, propagating along the trap axis, at an angle of 45$^{\circ}$ with respect to the quantization axis. For $k\in\{1,2\}$, the laser can be frequency modulated to be on resonant with the specific transition between the states $\left|0\right\rangle$ and $\left|k\right\rangle$ by an acousto-optic modulator (AOM). Coupling strength $\Omega_{k}$, duration $t$ and phase $\phi_{k}$ of this laser pulse then can be controlled by this AOM to perform a rotation, $R_k(\theta_k,\phi_k)$ between the states $\ket{0}$ and $\ket{k}$, where $\theta_k = \Omega_k t$ and $\phi_k$ represent the polar angle and the azimuthal angle respectively, i.e.,
\begin{equation}
R_1(\theta_1,\phi_1) =
\begin{pmatrix}
    \mathrm{cos} \frac{\theta_1}{2} & -i \mathrm{sin} \frac{\theta_1}{2} e^{-i \phi_1} & 0\\
    -i \mathrm{sin} \frac{\theta_1}{2} e^{i \phi_1} &  \mathrm{cos} \frac{\theta_1}{2} & 0\\
    0 & 0 & 1 \\
\end{pmatrix},
\end{equation}

\begin{equation}
R_2(\theta_2,\phi_2) =
\begin{pmatrix}
    \mathrm{cos} \frac{\theta_2}{2} & 0 & -i \mathrm{sin} \frac{\theta_2}{2} e^{-i \phi_2}\\
     0 & 1  & 0\\
    -i \mathrm{sin} \frac{\theta_2}{2} e^{i \phi_2} & 0 & \mathrm{cos} \frac{\theta_2}{2} \\
\end{pmatrix}.
\end{equation}
The parameters of the rotation are controlled by varying the duration and phase of the corresponding laser pulse under constant intensity via the acoustic-optic modulator with high fidelity. Using $R_1$ and $R_2$, one can perform an arbitrary rotation in the qutrit space. Measurements in our setup are performed by fluorescence detection, i.e., counting photons. These operations can be combined to perform various experiments nearly satisfying \Assuref{KCBSexclusivity}. We delineate this procedure for realising the KCBS configuration (see \Defref{idealKCBSconfig}) at the logical level, followed by experimental details corresponding to each step.

\subsection{The KCBS configuration \label{subsec:KCBSconfig}}

Using the notation introduced in \Assuref{KCBSexclusivity}, we denote the binary measurements of interest by $M_i:=(\Pi_{0|i},\Pi_{1|i})$ and the quantum state of our system by $\rho$. To realise the KCBS configuration in \Defref{idealKCBSconfig}, we require $\Pi_{1|i} = \ket{u_i}\bra{u_i}$ and $\rho = \ket{0} \bra{0}$. To implement the measurement $\Pi_{1|i}$, as a composition of the operations our setup permits, we first compute a qutrit rotation $U_i$ which maps $\ket{u_i}$ to $\ket{0}$. This, in turn, is used to compute the parameters $\theta_{1,i},\theta_{2,i}$ and $\phi_{1,i},\phi_{2,i}$ such that $U_i=R_2(\theta_{2,i},\phi_{2,i}) \cdot R_1(\theta_{1,i},\phi_{1,i})$. Observe that defining $\Pi_{1|i}:=U_i \ket{0}\bra{0} U_i^{\dagger}$ and $\Pi_{0|i}:=U_i (\mathbb{I}-\ket{0}\bra{0}) U_i^{\dagger}$ results in the KCBS configuration.%

If experimental imperfections are neglected (which are indeed negligible in our experiment; see \Subsecref{noiseComment}), it is evident that \Assuref{KCBSexclusivity} is satisfied, i.e., $\Pi_{1|i}\cdot \Pi_{1|j} = 0$ for $(i,j)\in E$ and $\Pi_{1|i}^2 = \Pi_{1|i}$. Further, it is clear that
\be
p_i = {\rm Pr}(1|i) = {\rm Pr}(10|ij) = \tr( \Pi_{0|j} \cdot \Pi_{1|i} \rho \Pi_{1|i}^{\dagger} \cdot \Pi_{0|i}^{\dagger}),
\ee
where $\Pr (ab|ij)$ denotes the probability of obtaining outcomes $(a,b)$ when measurements $(M_i,M_j)$ are performed and $\Pr (a|i)$ denotes the probability of obtaining outcome $a$ given $M_i$ was measured. Finally, we note that $\Pr(10|ij) = \Pr(01|ji)$.

As illustrated in \Figref{experiment_result}, $\sum_k p_k$ is estimated by randomly choosing (using a quantum random number generator; see \Subsecref{noiseComment}) $(i,j) \in E$, measuring $M_i$ followed by $M_j$ and counting $N(10|ij)$, the number of times the output was $(1,0)$, i.e.,
\be
p_i \approx \frac{N(10|ij)}{\sum_{a,b\in\{0,1\}} N(ab|ij)}.
\ee

\subsection{Experimental setup | Details  \label{subsec:experimental_setup}}

The Calcium ions, in which two dark states and one bright state are chosen to encode the states \{$\left|1\right\rangle $,$\left|2\right\rangle $\} and $\left|0\right\rangle $, respectively.
Every experimental sequence starts with 1 ${\rm ms}$ Doppler cooling using a 397 ${\rm nm}$ laser and a resonant 866 nm laser. The 397 nm laser is red detuned approximately half a natural linewidth from resonating with the cycling transition between $S_{1/2}$ and $P_{1/2}$ manifolds. This is followed by 300 $\mu s$ electromagnetically-induced transparency (EIT) cooling sequence which cools the ion down to near the motional ground state. After that, 10 $\mu s$ optical pumping laser is applied to initialize the qutrit to the $\left|0\right\rangle $ state with 99.5\% fidelity. When needed, 729 ${\rm nm}$ pulse sequence is applied to prepare a certain initial state. Measurements of the observables $\{A_{i}\}$ are performed through coherent rotations
$R_{2}(\theta_{2},\phi_{2})R_{1}(\theta_{1},\phi_{1})$, a fluorescence detection followed by the operation $R_{1}^{\dagger}(\theta_{1},\phi_{1})R_{2}^{\dagger}(\theta_{2},\phi_{2})$ which undoes the previous rotation. The first fluorescence detection lasts for 220 $\mu s$ and uses laser settings close to those of Doppler cooling, to minimize the motional heating due to photon recoil (in case of a bright state). %
However, this brings the temperature of the bright state close to the Doppler limit. 
Therefore, a 150 $\mu s$ EIT cooling sequence followed by 10 $\mu s$ optical pumping pulses is applied to cool the bright state ion.
The detection and cooling time are greatly suppressed to minimize the random phase accumulation between the sequential measurement. Neither fluorescence detection nor EIT cooling process affects the temperature of a dark state ion, which is slightly heated at a rate of 140 quanta per second due to electronic noise. The second fluorescence detection uses resonant laser pulse and 300 $\mu s$ integrating time, which lowers detection error.
The result is acquired by counting the number of photons collected from photo multiplier tube (PMT) within the detection time. The statistical probability is calculated by repeating the same measurement 10000 times.

In our experiments the $2\pi$ pulse time for both transitions are adjusted to around 142 $\mu s$, that is, the Rabi frequency is as low as $\Omega_{1/2}=(2\pi)7$ KHz, making the AC-Stark shift below 100 Hz. The separation between $\left|1\right\rangle $
and $\left|2\right\rangle $ is $\omega_{2}-\omega_{1}=(2\pi)8.69$ MHz with corresponding magnetic field $B=5.18$ G. The maximum probability of off-resonant excitation $\Omega^{2}/(\omega_{2}-\omega_{1})^{2}$ is about 6.5 $\times$ $10^{-7}$, small enough to ensure the independence of each Rabi oscillation \cite{Pan2019,Zhang2018}.

\subsection{The experiment \label{subsec:theExp}}
To study the self-testing property of device, we perform experiments which deviate from the standard KCBS scenario, in two qualitatively different ways.

First, we consider the realisation $(\rho'(p),\{M_i\})$ parametrised by $0\le p \le 1$. Here, $M_i$ correspond to projectors along $\ket{u_i}$ as before, but instead of initialising the state of the device to $\rho = \ket{0}\bra{0}$, we initialise it to a mixture of $\rho$ and the maximally mixed state, i.e.,
\begin{equation} \label{eq:mixed_state_p}
 \rho' (p) := (1-p) \ket{0}\bra{0} + \frac{p}{9} \mathbb{I}.
\end{equation}

For $p=0$, we recover the ideal KCBS configuration for which $\mathcal{I}=0$. However, for $p>0$, \Assuref{KCBSexclusivity} still holds (neglecting the imperfections in the measurements) but $\mathcal{I}=\epsilon>0$. The experiment was performed for $p=0,0.1$ and $0.2$.

Second, we consider the realisation $(\ket{u_0'}\bra{u_0'}, \{M_i'(\theta)\})$ parametrised by the vector $\ket{u_0'}$ and the angle $\theta$. Here, unlike the previous case, the state is pure. Further, the measurements $M_i' = (\Pi_{1|i}',\Pi_{0|i}')$ now correspond to projections on\footnote{Defined exactly as $M_i$ was defined using $\ket{u_i}$; $\Pi_{1|i}'=U_i' \ket{0}\bra{0} U_i^{\prime \dagger}$ and $\Pi_{0|i}'=U_i(\mathbb{I} - \ket{0}\bra{0})U_i^{\prime \dagger}$ with $U'_i$ encoding a rotation from $\ket{u'_i}$ to $\ket{0}$.} $\ket{u'_i}$ for $i\in\{1,2\dots 5\}$. These vectors $\ket{u'_i}$ are chosen to be
\begin{align}
  \left|u'_{1}\right\rangle &=(1,0,0)^{T}, \\
  \left|u'_{2}\right\rangle &=(0,0,1)^{T}, \\
  \left|u'_{3}\right\rangle &=(-\cos (\theta), \sin (\theta), 0)^{T}, \\
  \left|u'_{4}\right\rangle &=\frac{(\sin (\theta), \cos (\theta), \sin (\theta))^{T}}{\sqrt{1+\sin ^{2}(\theta)}},\\
  \left|u'_{5}\right\rangle &=(0, \sin (\theta),-\cos (\theta))^{T}  ,
\end{align}
which ensures that for all angles $\theta$, \Assuref{KCBSexclusivity} holds. We performed the experiment for two values of these parameters:
\begin{align}
  \{\theta=22.041,\quad &|u'_{0}\rangle=(0.686,-0.245,0.686)^{T}\},\text{ and} \label{eq:Second_point1}\\
  \{\theta=150.612,\quad &|u'_{0}\rangle=(-0.649,-0.400,-0.649)^{T}\}. \label{eq:Second_point2}
\end{align}

\subsection{The data\label{subsec:the_data}}

For the normal ordered case, the results are summarised in \Figref{experiment_result}. See Figure~\ref{fig:reverse_robustness_plot} in Appendix~\ref{sec:data} for the summary of reverse ordered experiments. Note that we have two robustness curves in both of the aforementioned figures. To see why we have two robustness curves, notice that the observed statistic constraint in the SDP corresponding to \Algref{The-algorithm} can be tightened by specifying the value of each $\left\langle \Pi_{i}\right\rangle$ term. Obviously, if one specifies the value of $\left\langle \Pi_{i}\right\rangle$, the value for their sum automatically gets fixed. Based on the experimental data, we have access to the value of each of the $\left\langle \Pi_{i}\right\rangle$ and thus we can always obtain a tighter bound compared with the bound obtained via the sum constraint in the SDP for \Algref{The-algorithm}. We illustrate this by plotting the robustness curve obtained in the special case $p_1=p_2 =p_3 =p_4 =p_5$. One could have just as easily taken a  different choice for the values of $p_i$ such that their sum remains equal to the appropriate KCBS value and obtain a different robustness curve. All such ``curves'' obtained by providing the value for $p_i$s will be lower bounded by the curve obtained via the sum constraint (the blue curve). For our experimental data point with mean $\mu$ and standard deviation $\sigma$, we base our analysis on $\mu-1.96\sigma$, which is the smallest possible value with $95$ percent confidence for a normally distributed data. %
The experimental data has been summarised in Table~\ref{tab:violation_table} and in Appendix~\ref{sec:data}. The complete dataset can be accessed on GitHub~\cite{swapKCBS}.  We now present the detailed analysis. %

\subsubsection{The first set of experiments}

For the first set of experiments, for $p=0, 0.1$ and $0.2$, we obtained $\sum_k p_k$ to be $2.233$, $2.186$ and $2.118$, respectively for the normal ordered case. By solving the SDP outputted by \Algref{The-algorithm} with these values as inputs, the lower bounds on the total fidelity of these realisations to the KCBS configuration is $5.296$, $3.002$ and $2.343$ respectively. For the reverse ordered experiments  corresponding to $p=0, 0.1$ and $0.2$, we obtained $\sum_k p_k$ to be $2.236$, $2.182$ and $2.124$, respectively. The lower bound on the total fidelity from the SDP  for these cases were $5.892$, $2.933$ and $2.386$ respectively.  As mentioned before, we base our analysis on $\mu-1.96\sigma$ for the appropriate values of  $\mu$ and $\sigma$. This, however, has consequences on the constraint that we impose via the observed statistics. For our experimental data, we observe that $\sum_{i} (\mu \left(p_i\right) -1.96 \sigma \left(p_i\right) )$ is  upper bounded by $\mu\left(\sum_i p_i\right)-1.96 \sigma\left(\sum_i p_i\right)$. Since our robustness curves are very sensitive to minor changes near the region of maximum KCBS value $\left(\sum_i p_i\right)$, the lower bound on the fidelity provided by the sum constraint can be better than the equal statistic  constraint ($p_1=p_2 =p_3 =p_4 =p_5$). Since our goal is to provide tight lower bounds, we take the largest value, which is the one obtained from sum constraint close to the region of maximum KCBS value. Away from the maximum KCBS value region,  equal statistic constraint provides a better estimate (as evident from the plots in \Figref{experiment_result}).   We also computed the fidelity of these quantum realisations to the ideal KCBS configuration by performing state and measurement tomography. This yielded $5.965$, $5.901$ and $5.812$ respectively. The values are same for normal and reverse ordered experiments bacause an experiment being normal or reverse is determined by  quantum random number generator and has no explicit difference otherwise at the level of measurements and the state involved. As expected, the experimentaly obtained total fidelity is lower bounded by the values obtained via \Algref{The-algorithm}.

\subsubsection{The second set of experiments}

For the second set of experiments (\Eqref{Second_point1} and \Eqref{Second_point2}), we obtained $\sum_k p_k$ to be $2.058$ and $2.043$ respectively, for the normal ordered case. For the reverse ordered experiments, the corresponding values were $2.057$ and $2.048$. While the SDP based lower bound on the total fidelity for the normal ordered case was $2.561$, $2.222$; the corresponding values for the reverse ordered case were $2.579$ and $2.579$. 

Finding the lower bound on the total fidelity is a bit involved in this case. Our goal is to find the total fidelity of the experimentally realized configuration with a configuration which is related to the KCBS configuration  via an isometry. In the previous case, this isometry was an identity matrix (a trivial isometry).
Since the optimal configuration
is unique up to an isometry, any configuration which is related to the
KCBS configuration via an isometry can be considered an optimal configuration.
Our goal is to calculate the total fidelity of the experimental realization of 
$(\left|u_{0}^{\theta}\right\rangle \left\langle u_{0}^{\theta}\right|,\text{\ensuremath{\left\{ \left|u_{i}^{\theta}\right\rangle \left\langle u_{i}^{\theta}\right|\right\} }}_{i=1}^{5})$, 
say 
$(\rho_{0}^{E},\left\{  M_{i}^{E}\right\}_{i=1}^{5})$ 
to an optimal configuration, say 
$(\left|u_{0}^{O}\right\rangle \left\langle u_{0}^{O}\right|,\text{\ensuremath{\left\{ \left|u_{i}^{O}\right\rangle \left\langle u_{i}^{O}\right|\right\} }}_{i=1}^{5})$.
Let us denote the fidelity of $\vert u_{i}^{E}\rangle$ to $\vert u_{i}^{O}\rangle$
as $F_{i}^{E,O}$.Here,  $F_i^{E,O}$ denotes the fidelity of the $i$th component of the experimental realisation to that of the ideal configuration. Similarly other fidelities are denoted by $F_{i}^{E,\theta}$
and $F_{i}^{\theta,O}$.
We know that fidelty is not a valid metric, and hence, we will not
be able to apply triangle inequality directly. However, we can use
trace distance as a proxy to relate the aforementioned three configurations. After some simple calculation, we obtain  the following expression for the lower bound
on the total fidelity $\left(\sum_{i=0}^{5}F_{i}^{E,O}\right)$ (see Appendix~\ref{sec:fid_triangle} for details)

\begin{equation}
\sum_{i=0}^{5}F_{i}^{E,I}\geq6-\left(\sum_{i=0}^{5}\sqrt{1-F_{i}^{E,\theta}}\right)-\left(\sum_{i=0}^{5}\sqrt{1-F_{i}^{\theta,O}}\right).\label{eq:configuration_lower_bound}
\end{equation}
Based on the experimental data, we can calculate $\left(\sum_{i=0}^{5}\sqrt{1-F_{i}^{E,\theta}}\right)$.
The second total fidelity expression $\left(\sum_{i=0}^{5}\sqrt{1-F_{i}^{\theta,O}}\right)$
can be evaluated using numerical optimization. The experimental data has been summarised in Table~\ref{tab:violation_table}.

\begin{table}[]
\begin{tabular}{p{0.12\textwidth}p{0.07\textwidth}p{0.07\textwidth}p{0.09\textwidth}p{0.1\textwidth}p{0.1\textwidth}p{0.1\textwidth}ll}
\hline
Scenario      & \multicolumn{3}{c}{$\sum p_k $}        & \multicolumn{3}{c}{Fidelity (tomography)}   & \multicolumn{2}{c}{Fidelity (lower bound) using}  \\ %
              & $\mu$              & $\sigma$              & $\mu-1.96\sigma$        & $\mu$       & $\sigma$       & $\mu-1.96\sigma$  &$(\sum{p_k})_{(\mu - 1.96 \sigma)}$ &$((p_k)_{(\mu-1.96\sigma)})_k$ \\ \hline
$p=0$, N    & 2.249             & 0.009           & 2.233                      & 5.965      &  0.016   &   5.933            & 5.296       & 4.170              \\ 
$p=0$, R    & 2.255             & 0.011           & 2.236                      & 5.965      &  0.016   &   5.933            & 5.892       & 4.026              \\ 
$p=0.1$, N & 2.207             & 0.011           & 2.186                      & 5.901      &  0.016   &   5.870            & 2.934       & 3.002              \\ 
$p=0.1$, R & 2.203             & 0.011           & 2.182                      & 5.901      &  0.016   &   5.870            & 2.842       & 2.933              \\ 
$p=0.2$, N  & 2.140             & 0.011           & 2.118                      & 5.812      &  0.019   &   5.775            & 1.654       & 2.343              \\ 
$p=0.2$, R  & 2.145             & 0.011           & 2.124                      & 5.812      &  0.019   &   5.775            & 1.753       & 2.386              \\ 
$\theta=22.021$, N     & 2.078             & 0.010           & 2.058                      & 3.989      & 0.017    &   3.956            & 0.740       & 2.561              \\ 
$\theta=22.041$, R     & 2.077             & 0.010           & 2.057                      & 3.989      &   0.017  &   3.956            & 0.726       & 2.579              \\ 
$\theta=150.612$, N    & 2.062             & 0.010           & 2.043                      & 4.050      &   0.017  &   4.017            & 0.533       & 2.222              \\ 
$\theta=150.612$, R    & 2.068             & 0.010           & 2.048                      & 4.050      &   0.017  &   4.017            & 0.601       & 2.579              \\ \hline
\end{tabular}
\caption{The above table shows a summary of our experimental and numerical results. Scenario refers to the various experimental settings discussed in \Subsecref{theExp}. In the scenario column, ``N'' and ``R'' represent normal and reverse order. The data in the last column shows the lower bound on the total fidelity based on semidefinite programming from \Algref{The-algorithm}. Note that we discuss $\mu-1.96  \sigma$, which is the smallest possible value with $95$ percent confidence for a normally distributed data. %
}
\label{tab:violation_table}
\end{table}
\subsection{\label{subsec:noiseComment}Quantification of noise}

We look at three sources of deviation from our assumption---KCBS orthogonality (stated as \Assuref{KCBSexclusivity}), Perfect Detection (no post-selection) and Random Selection of Measurements (the experimenter is not correlated with the state being measured).

We implicitly made the last two assumptions in our analysis. \emph{Perfect detection} was assumed because our measurements were projective and do not allow for events where the measurement itself fails. Allowing for such a possibility amounts to post-selection which in turn allows non-contextual (classical) models to also violate the classical bound and therefore no self-test guarantees can be made in that case. 

We also assumed that the actions of the observer and the (quantum) state are uncorrelated. If such correlations are allowed, then the aforementioned issue again nullifies any self-test guarantees. \emph{Random Selection of Measurements} helps address this concern to some extent.\footnote{We do not discuss this further as this has been widely studied in the Bell nonlocality setting.}

\textbf{KCBS orthogonality.} Our theoretical analysis only applies for experiments which satisfy the KCBS orthogonality condition, i.e., \Assuref{KCBSexclusivity}. To simplify the discussion, we use $\Pi_i$ to denote $\Pi_{1|i}$. We slightly abuse the notation and use  $\Pi_i$ to denote the measurement itself as well. 
While it is impossible to perfectly satisfy \Assuref{KCBSexclusivity} in practice, our experiment comes close. Since we assume our measurements $\Pi_i$ to be \emph{projective}, they satisfy $\Pi_i^2=\Pi_i$. This in particular entails \emph{repeatability} at the level of observed statistics, $\langle \Pi_{i}^{2}\rangle =\langle \Pi_{i} \rangle $, i.e., the same measurement repeated twice should give the same result for a fixed state, which we quantify by considering 
\be 
R_i := \Pr(00|ii) + \Pr(11|ii),
\ee
where $\Pr(ab|ij)$ is the probability of obtaining outcomes $(a,b)$ when first $\Pi_i$ is measured, followed by $\Pi_j$ as defined in \Subsecref{KCBSconfig}. Corresponding to various measurement settings in different experimental scenarios, the repeatability in our experiment has been tabulated as \Tabref{repeatability_48,repeatability_150,repeatability_22} in Appendix~\ref{sec:noiseAppendix}. We find that the mean repeatability of our measurements for the entire experimental dataset is $\mu = 0.994217$ with standard deviation $\sigma = 0.00015985$.

We also assumed that our measurements $\Pi_i$ are cyclically \emph{orthogonal}, i.e., $\Pi_i \cdot \Pi_j = 0$ whenever $(i,j)\in E$. At the level of experimental statistics, this entails that 
\begin{enumerate}
\item $\Pr\left(11\vert ij\right)=\Pr\left(11\vert ji\right)=0$. This follows
from the fact that $\left\langle \Pi_{i}\Pi_{j}\right\rangle =\left\langle \Pi_{j}\Pi_{i}\right\rangle =0.$
\item $\Pr\left(10\vert ij\right)=\Pr\left(01\vert ji\right)$. This follows
from the fact that $\left\langle \Pi_{i}\left(\mathbb{I}-\Pi_{j}\right)\right\rangle =\left\langle \left(\mathbb{\mathbb{I}}-\Pi_{j}\right)\Pi_{i}\right\rangle .$
\item $\Pr\left(01\vert ij\right)=\Pr\left(10\vert ji\right)$. This follows
from the fact that $\left\langle \Pi_{j}\left(\mathbb{I}-\Pi_{i}\right)\right\rangle =\left\langle \left(\mathbb{\mathbb{I}}-\Pi_{i}\right)\Pi_{j}\right\rangle .$
\item $\Pr\left(00\vert ij\right)=\Pr\left(00\vert ji\right)$. This follows
from the fact that $\left\langle \left(\mathbb{I}-\Pi_{i}\right)\left(\mathbb{I}-\Pi_{j}\right)\right\rangle =\left\langle \left(\mathbb{I}-\Pi_{i}\right)\left(\mathbb{I}-\Pi_{j}\right)\right\rangle .$
\end{enumerate}
To quantify the deviations from \emph{orthogonality} in the experimental
data for projectors $\Pi_{i}$and $\Pi_{j}$, we use the following two quantities:
\begin{equation}
\delta_{ij}=\sum_{a,b\in\left\{ 0,1\right\} }\left|\Pr\left(ab\vert ij\right)-\Pr\left(ba\vert ji\right)\right|,\label{eq:orthogonality_1}
\end{equation}
and 
\be
o_{ij}:=p\left(11\vert ij\right).
\ee
Deferring the individual estimates of $\delta_{ij}$ and $o_{ij}$ to \Tabref{orthogonality,compatibility} in Appendix~\ref{sec:noiseAppendix}, respectively, we report their mean and standard deviation averaged over all $i,j$: $\delta=0.016 \pm 0.012$ and $o = 0.000550 \pm 0.0000984$. %

\textbf{Perfect detection.} A quantum jump detection is used to detect the probability in $\left|S_{1/2}\right\rangle $, which is the ``bright state''. The photons that radiate from the ion are collected through two objectives each with numerical aperture around ${\rm NA}=0.32$. In this experiment, we use two different settings of the detection laser for the first and second photon detection steps. In the first detection step, the frequency of 397 ${\rm nm}$ laser is red detuned by half the natural linewidth and the laser intensity is set to lower level to minimize the heating effect, and the detection time is suppressed to 220 ${\rm \mu s}$ to minimize the random phase accumulation and decoherence. We observe on average 12 photons for the state $\left|0\right\rangle $ and less than 2.5 photons for the state $\left|1\right\rangle $ or $\left|2\right\rangle $.
That is to say, the threshold value to distinguish bright state and dark state is set to $n_{ph}=2.5$.
The state detection error rates for wrongly registering the state $|0\rangle$ and missing the state $|0\rangle$ are $0.07\%$ and $0.2\%$, respectively.
For the second detection step, the laser frequency is set to on resonance. Together with 6 times higher laser power and 300 ${\rm \mu s}$ detection time, detection error is further suppressed. We observe on average 41 photons for the state $\left|0\right\rangle $ and 0.3 photons for the state $\left|1\right\rangle $ or $\left|2\right\rangle $. By setting the threshold to $n_{ph}=8.5$, the state detection error rates for wrongly registering the state $|0\rangle$ and missing the state $|0\rangle$ are far less than $0.04\%$. Since the detection efficiency of the ion trap is close to $100\%$, our detection is nearly perfect.

\textbf{Random selection of measurements.} The choice of observables to measure is made randomly in our experiment using a quantum random number generator (QRNG). %
Before each measurement, we prepare the ions to the $|0\rangle$ state. First, a random integer $m$ between 1 and 10 is generated. The random number $i = m \mod 5$ selects an observable $\Pi_{i}$ from $\{\Pi_{1},\Pi_{2},\dots,\Pi_{5}\}$, with $\Pi_{0}$ equal to $\Pi_{5}$, Then, one of the two observables $\Pi_{i\pm1}$ which are compatible with $\Pi_{i}$ is selected randomly by the sign of $5.5-m$. The random integer from the RNG determining control sequences are updated in real time (pregenerated before every experiment).\\ %

\subsubsection{Relation with prior contextuality experiments} \label{subsubsec:prior_noise}
To the best of our knowledge, in all previous experiments on contextuality, an operational approach was taken. This approach is of interest because for foundational questions, one must phrase the necessary assumptions in a theory-independent language. Our experiment may also be charactarised in this language---demonstration of perfect detection addresses \emph{the detection loophole} while the random selection of measurements addresses \emph{the free choice loophole}. We already saw a brief description of these loopholes. We now state the analogue of our KCBS orthogonality condition. At the operational level, the notion of contextuality is usually defined for a set of \emph{sharp} measurements which satisfy some \emph{compatiblity} relation among them. We briefly review their definitions. 

\textbf{Sharp measurements} are repeatable and minimally disturbing---any measurement compatible with it should yield the same output distribution regardless of whether it was measured or not. It is easy to see that sharp measurements imply, in particular, that $R_i=1$ and $\delta_{ij}=0$. Our data therefore also serves as evidence for sharpness of measurements.

\textbf{Compatibility.} Two observables are compatible if they can be jointly measured. In case the observables are measured sequentially, this implies that the order in which they are measured does not matter. Again, it is easily seen that this would entail $\delta_{ij}=0$. Our data, as discussed above, supports the conclusion that our measurements are targeting compatible observables (corresponding to the five cycle graph).\footnote{One might wonder how compatibility and sharpness are different (aside from the repeatability aspect). The difference is that for sharpness, the non-disturbance condition must hold for all compatible observables. The compatibility condition may only require one to verify a subset of these observables to be compatible.}

\section{Outlook and Conclusion}

The lower bounds on total fidelity obtained using semi-definite programming only hold when \Assuref{KCBSexclusivity} is satisfied. Relaxing this assumption to allow for measurements  with errors $\epsilon$ in repeatability and $\delta$ in orthogonality is left as an open question. A full analysis would render the lower bound on the total fidelity as a function of the KCBS value, $\delta$ and $\epsilon$. Note that the experimental points in Figure ~\ref{fig:experiment_result} lie above our robustness curve. The robustness curve obtained via the full analysis will be below the present curve and thus will still lower bound the total fidelity for the experimental data points. We remark that the existing robustness curves in Bell settings only consider fidelity lower bound as a function of quantum value and hence a similar full analysis should include the errors emanating from detection and locality loopholes. A naïve application of current techniques fails because \Algref{ConstructingAnIsometry} no longer yields a valid isometry. Complementary to this, improving the sharpness and exclusivity of measurements in the experiment is also an obvious yet important goal, which if properly achieved, may render the aforesaid irrelevant.

In this work, we considered the simplest self-testing scenario. Extending the analysis to $n>5$ (by making the algorithm more efficient), adapting it to more involved settings (possibly to anti-cycles), dropping the IID assumption (perhaps by using martingale analysis), obtaining analytic bounds (possibly via techniques from convex optimisation), we believe, are interesting avenues for further exploration.

\section*{Acknowledgements}
This work was supported by the National Key Research and Development Program of China (No. 2017YFA0304100, No. 2021YFE0113100), NSFC (No.\ 11734015, No.\ 11874345, No.\ 11821404, No.\ 11904357, No.\ 12174367, No.\ 11904402, No.\ 12074433, No.\ 12174447, No.\ 12004430, and No.\ 12174448), the Fundamental Research Funds for the Central Universities, USTC Tang Scholarship, Science and Technological Fund of Anhui Province for Outstanding Youth (2008085J02). AC is supported by Project Qdisc (Project No.\ US-15097, Universidad e Sevilla), with FEDER funds, QuantERA grant SECRET, by MINECO (Project No.\ PCI2019-111885-2), and MICINN (Project No.\ PID2020-113738GB-I00). LCK thanks the Ministry of Education, Singapore and the National Research Foundation Singapore for their support. KB~acknowledges funding by AFOSR, DoE QSA, NSF QLCI (award No.~OMA-2120757), DoE ASCR Accelerated Research in Quantum Computing program (award No.~DE-SC0020312), NSF PFCQC program, the DoE ASCR Quantum Testbed Pathfinder program (award No.~DE-SC0019040), U.S. Department of Energy Award No.\ DE-SC0019449, ARO MURI, AFOSR MURI, and DARPA SAVaNT ADVENT. ASA acknowledges funding provided by the Institute for Quantum Information and Matter. A substantial part of the work was done while ASA was at the Université libre de Bruxelles and was supported by the Belgian Fonds pour la Formation à la Recherche dans l’Industrie 4959 et dans l’Agriculture - FRIA, under grant No.\ 1.E.081.17F. ASA and JR were supported by the Belgian
Fonds de la Recherche Scientifique – FNRS, under grant
no R.50.05.18.F (QuantAlgo). The QuantAlgo project has received funding from the QuantERA European Research Area Network (ERA-NET) Cofund in Quantum Technologies implemented within the European Union’s Horizon 2020 program. We would like to express our gratitude to Jonathan Lau for his logistical assistance with computer-related matters. We thank Tobias Haug for various discussions.

\bibliographystyle{apsrev4-1}
\bibliography{ST_Experiment}

\appendix

\section{Experimental datapoints to assess repeatability and orthogonality\label{sec:noiseAppendix}}
In this section, we  present the details for repeatability and orthogonality corresponding to various measurement settings in different experimental scenarios. See Tables~\ref{tab:repeatability_48},~\ref{tab:repeatability_150} and~\ref{tab:repeatability_22} for details regarding repeatability. For orthogonality, refer to Tables~\ref{tab:orthogonality} and~\ref{tab:compatibility}.

\begin{table}[H]
  \centering
\begin{tabular}{lll}
\hline
Repeatability & $\mu$  & $\sigma$    \\ \hline
$R_1$            & 0.994 & 0.001 \\ 
$R_2$             & 0.991 & 0.002 \\ 
$R_3$             & 0.990 & 0.002 \\ 
$R_4$             & 0.993 & 0.001 \\ 
$R_5$             & 0.994 & 0.001 \\ \hline
\end{tabular}
\caption{The above Table shows the repeatability of the measurements corresponding to the first set of experiments. For every data-point, we collect $14999$ samples.}
\label{tab:repeatability_48}
\end{table}

\begin{table}[H]
  \centering
\begin{tabular}{lll}
\hline
Repeatability & $\mu$  & $\sigma$    \\ \hline
$R_1$             & 0.999 & 0.002     \\ 
$R_2$             & 0.995 & 0.001 \\ 
$R_3$             & 0.996 & 0.001 \\ 
$R_4$             & 0.992 & 0.001 \\ 
$R_5$             & 0.992 & 0.002 \\ \hline
\end{tabular}
\caption{The above Table shows the repeatability of the measurements corresponding to the second set of experiments for $\theta = 150.612$. Zero standard deviation in some of the cases appear due to rounding up to three digits after decimal.  For every data-point, we collect $14999$ samples.}
\label{tab:repeatability_150}
\end{table}

\begin{table}[H]
  \centering
\begin{tabular}{lll}
\hline
Repeatability & $\mu$  & $\sigma$    \\ \hline
$R_1$             & 0.999 & 0.001 \\ 
$R_2$             & 0.996 & 0.001 \\ 
$R_3$             & 0.999 & 0.002 \\ 
$R_4$             & 0.992 & 0.001 \\ 
$R_5$             & 0.994 & 0.001 \\ \hline
\end{tabular}
\caption{The above Table shows the repeatability of the measurements corresponding to the second set of experiments for $\theta = 22.041$. Zero standard deviation in some of the cases appear due to rounding up to three digits after decimal.  For every data-point, we collect $14999$ samples.}
\label{tab:repeatability_22}
\end{table}

\begin{table}[H]
  \centering
\begin{tabular}{p{0.2\textwidth}p{0.15\textwidth}p{0.15\textwidth}}
\hline
Scenario      & $\mu= 10^{-2}\times$ & $\sigma = 10^{-2}\times$\\ \hline
$p=0$, N    &  0.500     & 0.024    \\ 
$p=0$, R     &  0.819    &  0.040   \\ 
$p=0.1$, N &  0.875    & 0.040   \\ 
$p=0.1$, R  &   0.802   &  0.038  \\ 
$p=0.2$, N  &   0.801   &  0.040  \\ 
$p=0.2$, R  &  0.667    &  0.353  \\ 
$\theta=22.041$, N     &   0.385   &   0.027 \\ 
$\theta=22.041$, R     &    0.305  &  0.024  \\ 
$\theta=150.612$, N    &   0.210   & 0.020  \\ 
$\theta=150.612$, R    &   0.198   & 0.019   \\ \hline
\end{tabular}
\caption{The above table summarizes the deviation from orthogonality (quantified via $o_{ij}$ for measurements $i$ and $j$) for various experimental scenario. We present the mean $\mu$ and standard deviation $\sigma$ for the set of values $\{o_{ij}\}_{i=1}^{5}$, where $j=i+1$. Here, ``scenario'' refers to various experimental settings as discussed in the paper. In the scenario column, ``N'' and ``R'' represent normal and reverse order.}
\label{tab:orthogonality}
\end{table}

\begin{table}[H]
  \centering
  \begin{tabular}{llllllll}
  \hline
  Scenario   & $\delta_{12}$   & $\delta_{23}$   & $\delta_{34}$   & $\delta_{45}$   & $\delta_{51}$   & $\mu$  & $\sigma$    \\ \hline
  $p= 0$ & 0.014 & 0.030 & 0.017 & 0.027 & 0.003 & 0.018 & 0.010 \\ 
  $p =0.1$& 0.019 & 0.014 & 0.017 & 0.013 & 0.010 & 0.015 & 0.003 \\ 
  $p= 0.2$& 0.008 & 0.011 & 0.023 & 0.018 & 0.023 & 0.017 & 0.006 \\ 
  $\theta=22 $        & 0.006 & 0.011 & 0.029 & 0.010 & 0.004 & 0.012 & 0.009 \\ 
  $\theta=150$        & 0.008 & 0.008 & 0.010 & 0.008 & 0.008 & 0.008 & 0.001 \\ 
  Overall    & 0.007 & 0.005 & 0.022 & 0.037 & 0.009 & 0.016 & 0.012 \\ \hline
  \end{tabular}
  \caption{The above table shows the possible values of $\delta_{ij}$ for different experimental scenario. Intuitively speaking, $\delta_{ij}$ captures the difference in statistics for normal and reverse order experiments.}
  \label{tab:compatibility}
  \end{table}

\section{State and measurement tomography}
In this section, we discuss the experimental process involved towards measurement and state tomography.
We prepare the initial state in the $\left|0\right\rangle $ state with 10 $\mu s$ optical pumping. In order to find the fidelity of
this initial state, we did state tomography according to the Ref. \cite{Thew2001}. As for a qutrit, we can write the density matrix as $\rho=\frac{1}{3}\sum_{j=0}^{8}r_{j}\lambda_{j}$, where the $\lambda_{j}$ are the SU(3) generators and identity operator $\lambda_{0}$. The measurement basis $\{\left|\psi_{j}\right\rangle \}$ used in state tomography are chosen to be $\left|0\right\rangle $,
$\left|1\right\rangle $, $\left|2\right\rangle $, $(\left|0\right\rangle +\left|1\right\rangle )/\sqrt{2}$,
$(\left|0\right\rangle +i\left|1\right\rangle )/\sqrt{2}$, $(\left|0\right\rangle +\left|2\right\rangle )/\sqrt{2}$,
$(\left|0\right\rangle +i\left|2\right\rangle )/\sqrt{2}$, $(\left|1\right\rangle +\left|2\right\rangle )/\sqrt{2}$,
$(\left|1\right\rangle +i\left|2\right\rangle )/\sqrt{2}$, and we
use the projection operator $\left|\psi_{j}\right\rangle \left\langle \psi_{j}\right|$ as the generators. All the results are shown in Tables~\ref{tab:state_fidelity} and ~\ref{tab:state_fidelity_2}.

We did measurement tomography according to Ref. \cite{Fiu2001}. The probability $p_{lm}$ that the apparatus will respond with positive
operator-values measure (POVM) $\Pi_{l}$ when measuring the quantum state with density matrix can be expressed as $p_{lm}={\rm tr}[\Pi_{l}\rho_{m}]$,
where tr stands for the trace. Assuming that the theoretical detection probability $p_{lm}$ can be replaced with relative experimental frequency
$f_{lm}$, we may write ${\rm Tr}[\Pi_{l}\rho_{m}]=\sum_{i,j=1}^{N}\Pi_{l,ij}\rho_{m,ji}=f_{lm}$, where $N$ is the dimension of Hilbert space on which the operators $\Pi_{l}$ act. Then we can reconstruct the POVM $\Pi_{l}$ according to experimental result $f_{lm}$. The density matrix $\rho_{m}$ chosen in this process are the same as projection operators $\left|\psi_{j}\right\rangle \left\langle \psi_{j}\right|$ in state tomography. Also the maximum likelihood estimation(MLE) method is used in this process to ensure $\Pi_{l}\geqslant0$ and $\sum_{l=1}^{k}\Pi_{l}=I$, where $k$ is the number of measurement outcomes. The results are shown in Tables~\ref{tab:Theta_48_meas_fidelity}, \ref{tab:Theta_22_meas_fidelity} and ~\ref{tab:Theta_150_meas_fidelity}. \\

\section{Fidelity lower bound in the second set of experiments} \label{sec:fid_triangle}
The trace distance between two matrices $A$ and $B$ is given by
$D(A,B)=\frac{1}{2}\left\Vert A-B\right\Vert _{tr}$, where $\left\Vert \right\Vert _{tr}$
denotes the trace norm. If $A$ and $B$ are density matrices, then
we have the following lower and upper bounds on the trace distance
in terms of fidelity,

\begin{equation}
1-\sqrt{F(A,B)}\leq D(A,B)\leq\sqrt{1-F(A,B).}\label{eq:bound_trace_distance}
\end{equation}
After applying traingle inequality for density matrices $A$, $B$
and $C$ gives

\begin{equation}
1-\sqrt{F(A,C)}\leq D(A,C)\leq D(A,B)+D(B,C)\leq\sqrt{1-F(A,B)}+\sqrt{1-F(B,C)}.\label{eq:traingle_inequality_fidelity}
\end{equation}
This gives a version of ``traingle inequality for fidelity,''

\begin{equation}
\sqrt{F(A,C)}\geq1-\sqrt{1-F(A,B)}-\sqrt{1-F(B,C)}.\label{eq:final_triangle_inequality_fidelity}
\end{equation}
If at least one of the states among $A$ and $C$ is a pure state,
one can tighten the lower bound on $D(A,C)$,

\begin{equation}
1-F(A,C)\leq D(A,C).\label{eq:pure-state_lower bound.}
\end{equation}
Thus, we have

\begin{equation}
F(A,C)\geq1-\sqrt{1-F(A,B)}-\sqrt{1-F(B,C)}.\label{eq:pure_state_fidelity_triangle_inequality}
\end{equation}
We use inequality \ref{eq:pure_state_fidelity_triangle_inequality}
to relate the total fidelity of the experimental configuration $( \rho^E, \left\{ \Pi_i^E \right\} )$ 
to an optimal configuration $( \ket{u_0^O}\bra{u_0^O}, \{ \ket{u_i^O} \bra{u_i^O} \}_i )$. 
Let us denote the fidelity of the experimental configuration to $\vert u_{i}^{O}\rangle$
as $F_{i}^{E,O}$. Similarly other fidelities are denoted by $F_{i}^{E,\theta}$
and $F_{i}^{\theta,O}$. It is easy to see the following lower bound
on the total fidelity $\left(\sum_{i=0}^{5}F_{i}^{E,O}\right)$:
\begin{equation}
\sum_{i=0}^{5}F_{i}^{E,O}\geq6-\left(\sum_{i=0}^{5}\sqrt{1-F_{i}^{E,\theta}}\right)-\left(\sum_{i=0}^{5}\sqrt{1-F_{i}^{\theta,O}}\right).\label{eq:configuration_lower_bound}
\end{equation}

\section{Further details about the experimental data}\label{sec:data}
In this section, we discuss further details regarding the experimental data. In~\Figref{reverse_robustness_plot}, we present the data corresponding to the reverse order experiments. The fidelity of the prepared quantum state to the state to be prepared  is shown in Tables~\ref{tab:state_fidelity} and ~\ref{tab:state_fidelity_2}. For the first set of experiments corresponding to $\theta=48$, the fidelity of the implemented  projector to the the projector to be implemented has been shown in Table~\ref{tab:Theta_48_meas_fidelity}. For the second set of experiments corresponding to (\Eqref{Second_point1} and \Eqref{Second_point2}), the corresponding fidelities has been shown in Tables~\ref{tab:Theta_22_meas_fidelity} and ~\ref{tab:Theta_150_meas_fidelity}.

\begin{figure}[htbp]
    \centering
    \includegraphics[width=0.4\paperwidth]{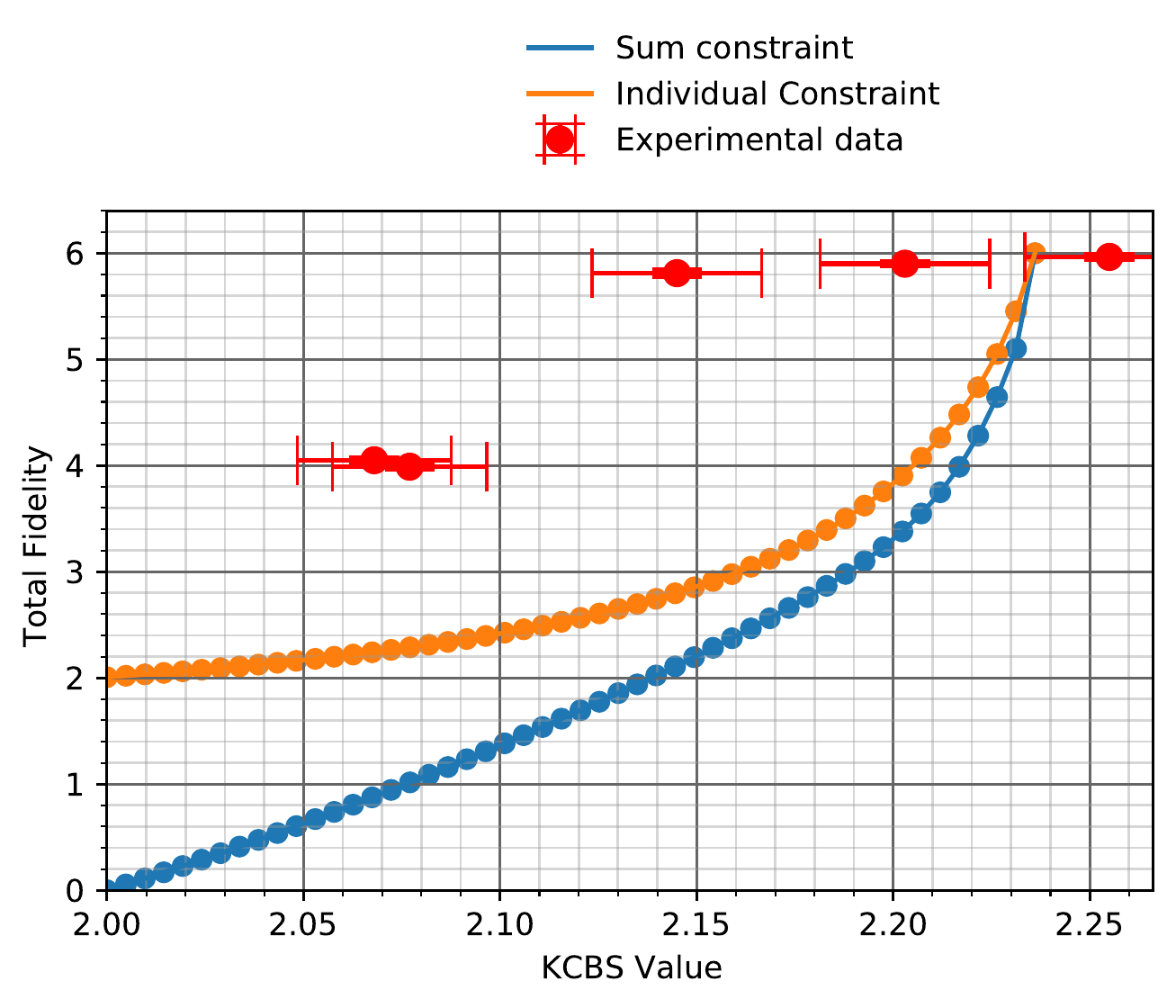}
    \caption{The above plot is the analog of ~\Figref{experiment_result} corresponding to the data obtained via reverse ordered experiments.}
    \label{fig:reverse_robustness_plot}
\end{figure}

\begin{table}[]
\begin{tabular}{lll}
\hline
State             & Fidelity to self & Fidelity to $\ket{0}$ \\ \hline
$\ket{0}$                 & 0.999 $\pm$0.008         & 0.999 $\pm$ 0.008      \\ 
$\ket{1}$                     & 0.999 $\pm$0.006         & 0.004 $\pm$ 0.002       \\ 
$\ket{2}$                     & 0.993 $\pm$ 0.012        & 0.005 $\pm$ 0.002      \\ 
Mixed with $p=0.1$ & 0.996 $\pm$ 0.010         & 0.935$\pm$ 0.008      \\ 
Mixed with $p=0.2$  & 0.992 $\pm$ 0.008         & 0.846 $\pm$ 0.013     \\ \hline
\end{tabular}
\caption{The above table shows the fidelity of the prepared state to the fidelity of the state to be prepared (the state $\ket{0}$) in row $1$,$4$ and $5$. Here, the state with mixedness $p$ represents the density matrix $ \rho' (p) = (1-p) \ket{0}\bra{0} + \frac{p}{9} \mathbb{I}.$ The numbers after $\pm$ represent the standard deviation  values. The data in row $2$ and $3$ is for sanity check, i.e, the state state $\ket{1}$ and $\ket{2}$ should be orthogonal to the state $\ket{0}.$ For every data-point, we collect $4000$ samples.}
\label{tab:state_fidelity}
\end{table}

\begin{table}[]
\begin{tabular}{lll}
\hline
State           & Fidelity & $\sigma$ \\ \hline
$\theta$ = 22.041   & 0.994    & 0.008 \\ 
$\theta$ = 150.612 & 0.995    & 0.009 \\ \hline
\end{tabular}
\caption{ The above table shows the fidelity of the experimentally prepared state to the state to be prepared corresponding to \Eqref{Second_point1} and \Eqref{Second_point2}. For every data-point, we collect $4000$ samples.}
\label{tab:state_fidelity_2}
\end{table}

\begin{table}[]
\begin{tabular}{lll}
\hline
Measurements & Fidelity & $\sigma$ \\ \hline
$\Pi_1$         & 0.993    & 0.006 \\ 
$\Pi_2$         & 0.991    & 0.006 \\ 
$\Pi_3$         & 0.996    & 0.004 \\ 
$\Pi_4$         & 0.995    & 0.006 \\ 
$\Pi_5$         & 0.990    & 0.009 \\ \hline
\end{tabular}
\caption{The above table shows the fidelity of the experimentally implemented projectors to the projectors to be implemented corresponding to the first set of experiments. For every data-point, we collect $15000$ samples.}
\label{tab:Theta_48_meas_fidelity}
\end{table}

\begin{table}[]
\begin{tabular}{lll}
\hline
Measurements & Fidelity & $\sigma$ \\ \hline
$\Pi_1$         & 0.997    & 0.004 \\ 
$\Pi_2$         & 0.998    & 0.005 \\ 
$\Pi_3$         & 0.997    & 0.005 \\ 
$\Pi_4$         & 1.000    & 0.007 \\ 
$\Pi_5$         & 0.997    & 0.006 \\ \hline
\end{tabular}
\caption{The above table shows the fidelity of the experimentally implemented projectors to the projectors to be implemented corresponding to the second set of experiments for $\theta = 22.041$. For every data-point, we collect $15000$ samples.}
\label{tab:Theta_22_meas_fidelity}
\end{table}

\begin{table}[]
\begin{tabular}{lll}
\hline
Measurements & Fidelity & $\sigma$ \\ \hline
$\Pi_1$         & 0.995    & 0.004 \\ 
$\Pi_2$         & 0.995    & 0.007 \\ 
$\Pi_3$         & 0.994    & 0.005 \\ 
$\Pi_4$         & 0.996    & 0.005 \\ 
$\Pi_5$         & 0.993    & 0.005 \\ \hline
\end{tabular}
\caption{The above table shows the fidelity of the experimentally implemented projectors to the projectors to be implemented corresponding to the second set of experiments for $\theta = 150.612$. For every data-point, we collect $15000$ samples.}
\label{tab:Theta_150_meas_fidelity}
\end{table}

\end{document}